\documentclass[11pt,a4paper]{article}
\usepackage{amsfonts,amssymb,amsthm}
\usepackage{cite} 
\usepackage[T2A]{fontenc}
\usepackage[cp1251]{inputenc}
\usepackage[english,russian]{babel}
\usepackage{amsmath}%
\usepackage[unicode]{hyperref}
\usepackage[left=1.5cm,right=1.5cm,
    top=2cm,bottom=3cm,bindingoffset=0cm]{geometry}

\newtheorem{lemma}{Лемма}[section]
\newtheorem{theorem}{Теорема}[section]

\newtheorem{definition}{Определение}[section]
\newtheorem{proposition}{Предложение}[section]

\def\const {\hbox{const}}

\allowdisplaybreaks
\numberwithin{equation}{section}

\title{\bf О классификации нелинейных интегрируемых трехмерных цепочек при помощи характеристических алгебр Ли}
\author{\bf И.Т.~Хабибуллин, А.Р.~Хакимова }

\date{} 

\begin{document}
\maketitle



\abstract {В статье продолжается работа по описанию интегрируемых нелинейных цепочек с тремя независимыми переменными следующего вида  $u^j_{n+1,x}=u^j_{n,x}+f(u^{j+1}_{n},u^{j}_n,u^j_{n+1 },u^{j-1}_{n+1})$ по признаку наличия иерархии редукций,  интегрируемых в смысле Дарбу, начатая в \cite{KHKh}. В основе классификационного алгоритма лежит хорошо известный факт, что характеристические алгебры интегрируемых по Дарбу систем имеют конечную размерность. В работе мы пользовались характеристической алгеброй по направлению $x$, структура которой для данного класса моделей определяется некоторым полиномом $P(\lambda)$, степень которого для известных примеров не превосходит трех. В статье предполагается, что $P(\lambda)=\lambda^2$, в этом случае классификационная задача сводится к отысканию восьми неизвестных функций одной переменной. В работе получен достаточно узкий класс претендентов на интегрируемость, среди которых имеется новый пример интегрируемой цепочки.} 

\vspace{0.5cm}

\textbf{Ключевые слова:} {трехмерные цепочки, характеристические алгебры, интегрируемость по Дарбу, характеристические интегралы, интегрируемые редукции.}

\vspace{0.5cm}

\large

\section{Введение}

Теория интегрируемости является важной составляющей современной математической физики. К настоящему времени задача исчерпывающего описания интегрируемых представителей для широкого класса нелинейных непрерывных и дискретных моделей размерности 1+1, наиболее интересных с точки зрения приложений, близка к завершению. При решении этой задачи успешно использовались классификационные алгоритмы, основанные на симметрийном подходе \cite{ISh}, \cite{ZhSh}, \cite{SS}, \cite{MShY}, \cite{AShY}, \cite{Sokolov}. Для классификации уравнений с тремя и более независимыми переменными метод симметрий не подходит из-за проблем с нелокальностями. 
Кратко обсудим наиболее популярные методы классификации интегрируемых моделей в 3D. Метод классификации трехмерных уравнений, основанный на существовании гидродинамических редукций является одним из наиболее распространенных среди специалистов по математической физике. Он создавался, апробировался и использовался в серии работ  \cite{GK}, \cite{GT}, \cite{FMN}, \cite{SO},   \cite{Novikov}, \cite{FNR}. Метод классификации трехмерных интегрируемых цепочек, использующий идею совместности был разработан в статьях \cite{ABS1}, \cite{ABS2}. Существуют и альтернативные подходы к изучению и классификации уравнений в 3D (см., например, работы \cite{Pavlov}, \cite{BK}, \cite{CK} и ссылки в них).

В настоящей работе продолжается исследование задачи классификации нелинейных цепочек вида 
\begin{equation}\label{jnx-part}
u^j_{n+1,x}=u^j_{n,x}+f(u^{j+1}_{n},u^{j}_n,u^j_{n+1 },u^{j-1}_{n+1}),\quad \frac{\partial f}{\partial u^{j+1}_{n}}\neq0,\quad \frac{\partial f}{\partial u^{j-1}_{n+1}}\neq0 
\end{equation}
линейных по производным $u^j_{n,x}=\frac{d}{dx}u^j_n$, начатое в \cite{KHKh}. В уравнении \eqref{jnx-part} искомая функция $u^j_n(x)$ зависит от трех аргументов: непрерывной переменной $x$ и дискретных~$n$~и~$j$. Здесь $f$ предполагается аналитической функцией, заданной в некоторой области пространства \textbf{$C^4$}.  

Наш подход к задаче классификации интегрируемых цепочек в 3D основан на следующем любопытном наблюдении. Известные примеры интегрируемых уравнений с тремя независимыми переменными, когда хотя бы одна из переменных является дискретной, допускают редукции в виде интегрируемых в смысле Дарбу систем уравнений гиперболического типа с двумя независимыми переменными. Эти редукции получаются в результате наложения граничных условий типа обрыва по дискретному аргументу (см., например, \cite{Hab13}, \cite{KuznetsovaHabibullin18}, \cite{HabKh21}).
В работах \cite{KuznetsovaHabibullin20}, \cite{Ferapontov} такое свойство интегрируемых трехмерных цепочек было успешно использовано при классификации цепочек с одной дискретной и двумя непрерывными переменными. 

В работе \cite{KHKh} мы адаптировали алгебраический метод классификации на случай цепочек вида \eqref{jnx-part} с одной непрерывной и двумя дискретными независимыми переменными.	Мы предполагаем, что существуют функции $f^{-N_2}$ и $f^{N_1}$ такие, что система дифференциально-разностных уравнений гиперболического типа следующего вида 
\begin{equation}\label{sys}
u^j_{n+1,x}=u^j_{n,x}+f^{j}_n, \qquad -N_2\leq j\leq N_1, 
\end{equation}
где 
\begin{align*}
&f^{-N_2}_n=f^{-N_2}(u^{-N_2+1}_{n},u^{-N_2}_n,u^{-N_2}_{n+1 }),\\
&f^j_{n}=f(u^{j+1}_{n},u^{j}_n,u^j_{n+1 },u^{j-1}_{n+1}), \quad -N_2+1\leq j\leq N_1-1,\\
&f^{N_1}_n=f^{N_1}(u^{N_1}_n,u^{N_1}_{n+1 }, u^{N_1-1}_{n+1}),
\end{align*}
будет интегрируемой по Дарбу для любой пары неотрицательных чисел $N_1, N_2$. Тогда цепочки \eqref{jnx-part}, соответствующие найденным функциям $f=f(u^{j+1}_{n},u^{j}_n,u^j_{n+1 },u^{j-1}_{n+1})$, предположительно будут интегрируемыми. 

Напомним, что система уравнений \eqref{sys} является интегрируемой в смысле Дарбу, если она допускает полные наборы характеристических интегралов по каждому из характеристических направлений $x$ и $n$. Понятие характеристического интеграла дифференциального уравнения гиперболического типа впервые было введено в работе Дарбу \cite{Darboux}. Эффективный алгебраический критерий интегрируемости системы гиперболических уравнений был предложен А.Б.~Шабатом: конечномерность характеристической алгебры по каждому из характеристических направлений является необходимым и достаточным условием интегрируемости в смысле Дарбу системы гиперболических уравнений (см. \cite{Yamilov}, \cite{LSSh}, \cite{ZMHS}).
Теория интегрируемости по Дарбу полностью переносится на системы дифференциально-разностных и чисто дискретных уравнений \cite{Adler}, \cite{HPZh2008}, \cite{HPZh2009}, \cite{Smirnov}, \cite{ZhiberKuznetsova21}, \cite{HabKh22}.

В работе \cite{KHKh} было доказано, что необходимым условием конечномерности алгебры $L_x$ системы \eqref{sys} по направлению $x$ является выполнение равенства
\begin{equation}\label{eqs}
P(Z_k)f(u^{1}_{n},u^{0}_n,u^0_{n+1 },u^{-1}_{n+1})=0,\quad \mbox{для}\quad k=-1,0,1,
\end{equation}
где операторы $Z_k$ имеют следующий вид
$$Z_{-1}=\frac{\partial}{\partial u^{-1}_n},\quad Z_{0}=\frac{\partial}{\partial u^{0}_n}+ \frac{\partial}{\partial u^{0}_{n+1}},\quad Z_{1}=\frac{\partial}{\partial u^{1}_n},$$
а $P(\lambda)$ некоторый многочлен с постоянными коэффициентами
\begin{equation}\label{pol}
P(\lambda)=\lambda^M+a_1\lambda^{M-1}+ \dots +a_{M-1}\lambda
\end{equation}
с нулевым свободным членом. Далее предполагается, что $P(\lambda)$ многочлен наименьшей степени, удовлетворяющий условиям  \eqref{eqs}.

Задача об описании класса функций $f$ для которых система \eqref{sys} является интегрируемой по Дарбу естественным образом распадается на четыре варианта
\begin{enumerate}
	\item  $P(\lambda)=\lambda^2+a_{1}\lambda,\quad a_1\neq0$,
	\item  $P(\lambda)=\lambda^2,$
	\item  $P(\lambda)=\lambda^3+a_{1}\lambda^2+a_{2}\lambda,$
	\item  $\deg P(\lambda)\geq4$.
\end{enumerate}
Первый из перечисленных вариантов был исследован в работе \cite{KHKh}, где показано,  что для интегрируемости цепочки \eqref{jnx-part} в этом случае необходимо, чтобы она имела вид
\begin{align*}
u^j_{n+1,x}=u^j_{n,x}+k_1e^{u^{j-1}_{n+1}}+ e^{u^j_{n}}+ e^{u^{j}_{n+1}}+ k_2e^{u^{j+1}_{n}},
\end{align*}
где $k_1$, $k_2$ некоторые постоянные. Причем когда  $k_1=k_2=-1$ получаем известную интегрируемую цепочку, найденную ранее в работе \cite{Adler}. 

В настоящей работе мы исследуем случай 2. Случаи 3 и 4 остаются неизученными с точки зрения классификации по признаку интегрируемых по Дарбу редукций. Несколько примеров интегрируемых цепочек соответствующих случаю 3 можно найти в работе \cite{FNR}.
Примеры цепочек с функцией $f$, соответствующей пункту 4 авторам неизвестны. 

Основной результат работы представлен в Теоремах~\ref{final} и \ref{3.2}. В Теореме~\ref{final} утверждается, что если цепочка \eqref{jnx-part} с функцией $f$,  заданной в виде \eqref{form-f} интегрируема в смысле определения \ref{def}, то $f$ принадлежит одному из трех классов, содержащих неопределенные константы. В \S 3 детально обсуждается один из этих классов, в котором выделяются два подкласса. Представители одного из подклассов зависят от двух целочисленных параметров, в то время как цепочки из второго подкласса зависят от одного комплексного параметра (см. Теорема~\ref{3.2}). В \S 4 предъявлены три конкретные цепочки, для двух из них показано, что характеристические алгебры по направлению $x$ имеют конечные размерности, третья цепочка является новым интегрируемым примером. Для завершения классификации необходимо исследовать полученные цепочки с помощью характеристической алгебры по дискретному направлению~$n$.

\section{Классификационная задача}

Нетрудно проверить, что общее решение $f$ системы \eqref{eqs} при условии, что $P(\lambda)=\lambda^2$ представимо в виде  
\begin{align}\label{form-f}
\begin{aligned}
f\left(u^1_n, u^0_n, u^{0}_{n+1}, u^{-1}_{n+1}\right)=&A_0u^1_nu^0_nu^{-1}_{n+1}+A_1u^0_nu^{-1}_{n+1}+A_2u^1_nu^0_n+A_3u^1_nu^{-1}_{n+1}\\
&+A_4u^1_n+A_5u^0_n+A_6u^{-1}_{n+1}+A_7,
\end{aligned}
\end{align}
где коэффициенты $A_i=A_i(\tau_n)$ являются произвольными функциями от переменной $\tau_n=u^0_n-u^{0}_{n+1}$. Основная цель работы состоит в выводе необходимых условий на эти функции из предположения о конечномерности характеристической алгебры $L_x$ системы \eqref{sys} по направлению $x$. Ниже мы обсудим определение и некоторые свойства этой алгебры.

\subsection{Характеристическая алгебра по направлению $x$}

Напомним, что характеристическая алгебра $L_x$ системы \eqref{sys} по направлению $x$ определяется как алгебра Ли-Райнхарта над кольцом локально-аналитических функций  от динамических переменных, порожденная набором векторных полей, называемых характеристическими операторами (см. \cite{KHKh}):
\begin{align}
X_j=&\sum_{k=-\infty}^{\infty}\frac{\partial}{\partial u^j_{n+k}}, \quad j=-N_2,-N_2+1,\ldots,N_1, \label{Xj} \\
Y=&\sum_{j=-N_2}^{N_1}\left(f^j_{n}\frac{\partial}{\partial u^j_{n+1}}-f^j_{n-1}\frac{\partial}{\partial u^j_{n-1}}\right. \nonumber \\
&\left.+\left(f^j_{n}+f^j_{n+1}\right)\frac{\partial}{\partial u^j_{n+2}}-\left(f^j_{n-1}+f^j_{n-2}\right)\frac{\partial}{\partial u^j_{n-2}}+\ldots\right). \label{Y}
\end{align}
Отображение, действующее по правилу 
\begin{align}\label{auto-main}
Z\mapsto D_nZD^{-1}_n,
\end{align}
является автоморфизмом алгебры $L_x$ (см. \cite{ZMHS}). Для образующих алгебры имеем равенства 
\begin{align}\label{auto-Xj-Y}
D_nX_jD^{-1}_n=X_j, \qquad D_nYD^{-1}_n=Y-\sum_{j=-N_2}^{N_1}f^j_nX_j.
\end{align}

Условие конечномерности алгебры $L_x$ является необходимым и достаточным условием существования полного набора $x$-интегралов системы \eqref{sys}. Поэтому для того, чтобы система \eqref{sys} была интегрируема в смысле Дарбу необходимо, чтобы алгебра $L_x$ имела конечную размерность. Этот факт лежит в основе классификационного алгоритма, которым мы пользуемся при уточнении вида функции \eqref{form-f}.   

Для того, чтобы цепочка \eqref{jnx-part} с функцией $f$,  заданной в виде \eqref{form-f} допускала вырожденные обрывы с обеих сторон необходимо, чтобы она имела стационарное решение вида $u^j_n(x)=0$.  Это накладывает определенные ограничения на коэффициенты  $A_i(\tau_n)$, они должны быть непрерывны в нуле и, кроме того, требуется, чтобы $A_i(0)=0$. 

С учетом этого предположения система \eqref{sys} принимает вид
\begin{equation}\label{r-sys}
\begin{aligned}
&u^{N_1}_{n+1,x}=u^{N_1}_{n,x}+A_1(\tau^{N_1}_n)u^{N_1}_nu^{N_1-1}_{n+1}+A_5(\tau^{N_1}_n)u^{N_1}_n+A_6(\tau^{N_1}_n)u^{N_1-1}_{n+1}+A_7(\tau^{N_1}_n),\\
&u^{j}_{n+1,x}=u^{j}_{n,x}+f^j_n,\\
&u^{-N_2}_{n+1,x}=u^{-N_2}_{n,x}+A_2(\tau^{-N_2}_n)u^{-N_2+1}_nu^{-N_2}_{n}\\
&\qquad +A_4(\tau^{-N_2}_n)u^{-N_2+1}_n+A_5(\tau^{-N_2}_n)u^{-N_2}_{n}+A_7(\tau^{-N_2}_n),
\end{aligned}
\end{equation}
где $f^j=f\left(u^{j+1}_n, u^j_n, u^j_{n+1}, u^{j-1}_{n+1}\right)$, а функция $f$ задана в \eqref{form-f}. Ниже мы будем пользоваться следующим определением.
\begin{definition}\label{def}
Будем называть цепочку \eqref{jnx-part}, \eqref{form-f} интегрируемой, если характеристические алгебры системы \eqref{r-sys} по направлениям $n$ и $x$ имеют конечную размерность для любого выбора $N_1\geq 0$ и $N_2\geq 0$.
\end{definition}

\subsection{Поиск функции $A_0$}

Цель настоящего раздела состоит в определении функции $A_0(\tau)$. Имеет место следующая 
\begin{theorem}\label{ThA0}
Если $\dim L_x < \infty$, то функция  $A_0(\tau)$ имеет вид
\begin{equation}\label{A0}
A_0(\tau)=c_0\tau,
\end{equation}
где $c_0$ -- некоторая постоянная.
\end{theorem}
\begin{proof}
Сосредоточимся на подалгебре алгебры $L_x$ порожденной операторами\footnote{Здесь мы воспользуемся схемой предложенной ранее в \cite{HPZh2008}} 
\begin{equation}\label{T0T-A0}
T_0:=X_0, \qquad T=\left[X_1,\left[X_{-1},Y\right]\right].
\end{equation}
Рассмотрим последовательность кратных коммутаторов этих операторов:
\begin{equation}\label{seq-T0T-A0}
T_1=\left[T,T_0\right], \quad T_2=\left[T,T_1\right], \quad\ldots, \quad  T_{k+1}=\left[T,T_k\right], \quad k\geq0.
\end{equation}
Автоморфизм \eqref{auto-main} алгебры $L_x$ является полезным инструментом при изучении последовательностей. Пользуясь формулами \eqref{auto-Xj-Y} легко вывести следующие соотношения 
\begin{align*}
D_nTD^{-1}_n=T-\left(A_0u_n^0+A_3\right)T_0,\qquad
D_nT_1D^{-1}_n=T_1+A_0T_0.
\end{align*}
Можно проверить, что для $k\geq 2$ оператор $D_nT_kD^{-1}_n$ выражается не только через элементы последовательности \eqref{seq-T0T-A0}, но и через операторы, полученные применением степеней оператора $ad_{T_0}$, действующего по правилу $ad_{T_0}Z=[T_0,Z]$. Поэтому мы расширим последовательность \eqref{seq-T0T-A0} следующим образом. Введем мультииндекс $\alpha=(k,0,i_1,i_2,\ldots,i_{r-1},i_r)$, где $k$ некоторое натуральное число, переменная $i_j$ для $\forall j$ принимает одно из двух значений $0$ или $1$. Мультииндекс присваивается оператору согласно формуле
\begin{equation}\label{Talpha-A0}
T_\alpha= 
\begin{cases}
\left[T_0,T_{k,0,i_1,i_2,\ldots,i_{r-1}}\right], &\mbox{если}\quad i_r=0,\\
\left[T,T_{k,0,i_1,i_2,\ldots,i_{r-1}}\right], &\mbox{если}\quad i_r=1.
 \end{cases}
\end{equation}
Мы частично упорядочим множество операторов, пронумерованных мультииндексом, введением функции
\begin{equation}\label{malpha-A0}
m(\alpha)= 
\begin{cases}
k, &\mbox{если}\quad \alpha=k,\\
k, &\mbox{если}\quad \alpha=k,0,\\
k+i_1+i_2+\cdots+i_r, &\mbox{если}\quad \alpha=(k,0,i_1,i_2,\ldots,i_{r-1},i_r).
 \end{cases}
\end{equation}
Опишем действие автоморфизма \eqref{auto-main} на элементы последовательности операторов \eqref{Talpha-A0}. Легко можно проверить следующую формулу:
\begin{equation*}
D_nT_kD^{-1}_n=T_k+A_0T_{k-1}-\left(A_0u_n^0+A_3\right)\sum_{m(\beta)=k-1}T_{\beta}+\sum_{m(\beta)\leq k-2}\eta(k,\beta)T_{\beta}, \quad k\geq2.
\end{equation*}
Для операторов $T_\alpha$ не принадлежащих последовательности \eqref{seq-T0T-A0} имеем
\begin{equation*}
D_nT_\alpha D^{-1}_n=T_\alpha+\sum_{m(\beta)\leq m(\alpha)-1}\eta(\alpha,\beta)T_{\beta}.
\end{equation*}
Построим базис $P$ в линейном пространстве, порожденном элементами последовательности кратных коммутаторов операторов $T_0$, $T$. Для этого выполним следующие действия.
\begin{enumerate}
	\item[1)] Включим линейно независимые операторы $T_0$ и $T$ в базис $P$.
	\item[2)] Если оператор $T_1=\left[T,T_0\right]$ не зависит от $T_0$ и $T$, то включаем его в $P$. В противном случае $P$ состоит из двух операторов: $T_0$ и $T$.
	\item[3)] Предполагая, что $T_0, T, T_1 \in P$ изучим на предмет линейной зависимости четверку операторов $\left\{T_0, T, T_1, T_2\right\}$. Если операторы линейно независимы, то включим в $P$ все векторные поля $T_\beta$, где $m(\beta)=2$, $\beta\in I_2$ (в действительности, $I_2$ является набором таких $\beta$), такие, что набор операторов $J_2=\left\{T_0, T, T_1, T_2, \cup_{\beta\in I_2}T_\beta\right\}$ является линейно независимым, причем произвольное $T_\gamma$, $m(\gamma)=2$ является линейной комбинацией операторов из $J_2$.
	\item[4)] Далее проверяем является ли множество $T_3\cup J_2$ линейно независимой системой. Если нет, то не включаем в $P$ оператор $T_3$, а значит и все операторы $T_\beta$ с мерой $m(\beta)=3$. Если да, то мы включаем в $P$ все операторы $T_\beta$, где $m(\beta)=3$, $\beta\in I_3$ такие, что набор операторов $J_3=\cup_{m(\beta)=3, \,\beta\in I_3}T_\beta\cup\left\{T_3\right\}\cup J_2$ является линейно независимым множеством, причем для $\forall \gamma$, $m(\gamma)=3$ оператор $T_\gamma$ представим в виде линейной комбинации операторов из $J_3$.
\end{enumerate}
Продолжая этот процесс получим полный базис $P$ в линейном пространстве кратных коммутаторов.

В силу конечномерности рассматриваемой алгебры Ли-Райнхарта, начиная с некоторого натурального $N$ оператор $T_{N+1}$ линейно выражается через операторы $T_\beta$ для которых $m(\beta)\leq N$:
\begin{equation}\label{TN1-A0}
T_{N+1}=\mu(N+1,N)T_N+\sum_{T_\beta\in P,\, m(\beta)\leq N}\mu(N+1,\beta)T_\beta.
\end{equation}
По построению для $T_\gamma \notin P$, $m(\gamma)\leq N$ имеем также представление 
\begin{equation*}
T_\gamma=\sum_{T_\beta\in P,\, m(\beta)\leq m(\gamma)}\mu(\gamma,\beta)T_\beta.
\end{equation*}
Далее нам понадобится следующая лемма.
\begin{lemma}
Для любого векторного поля $T_\alpha \notin P$, $m(\alpha)=N$ в представлении 
\begin{equation*}
T_\alpha=\mu(\alpha,N)T_N+\sum_{T_\beta\in P}\mu(\alpha,\beta)T_\beta
\end{equation*}
коэффициент $\mu(\alpha,N)$ является константой.
\end{lemma}
\begin{proof}
Применим к последнему равенству оператор сопряжения и получим
\begin{equation*}
D_nT_\alpha D^{-1}_n=T_\alpha+\sum_{m(\beta)\leq N-1}\eta(\alpha,\beta)T_{\beta}
=\mu(\alpha,N)T_N+\sum_{T_\beta\in P}\mu(\alpha,\beta)T_\beta+\sum_{m(\beta)\leq N-1}\eta(\alpha,\beta)T_{\beta}.
\end{equation*}
С другой стороны имеем:
\begin{align*}
D_nT_\alpha D^{-1}_n&=D_n\left(\mu(\alpha,N)\right)D_nT_N D^{-1}_n+\sum_{T_\beta\in P}D_n\left(\mu(\alpha,\beta)\right)D_nT_\beta D^{-1}_n=\\
&=D_n\left(\mu(\alpha,N)\right)\left(T_N+\cdots\right)+\sum_{T_\beta\in P}D_n\left(\mu(\alpha,\beta)\right)\left(T_\beta+\cdots\right).
\end{align*}
Сравнивая в этих двух представлениях коэффициенты при $T_N$ находим $\mu(\alpha,N)=D_n\left(\mu(\alpha,N)\right)$ из которого следует, что $\mu(\alpha,N)$ является постоянной. Лемма доказана.
\end{proof}
Продолжим доказательство Теоремы~\ref{ThA0}. Подействуем на разложение \eqref{TN1-A0} автоморфизмом \eqref{auto-main} и получим слева
\begin{align*}
T_{N+1}&+A_0T_{N}-\left(A_0u_n^0+A_3\right)\sum_{m(\beta)=N}T_{\beta}+\sum_{m(\beta)\leq N-1}\eta(N+1,\beta)T_{\beta}=\\
&=\mu(N+1,N)T_N+\sum_{T_\beta\in P,\, m(\beta)\leq N}\mu(N+1,\beta)T_\beta+A_0T_{N}-\\
&-\left(A_0u_n^0+A_3\right)\sum_{m(\beta)=N}T_{\beta}+\sum_{m(\beta)\leq N-1}\eta(N+1,\beta)T_{\beta},
\end{align*}
соответственно, справа
\begin{align*}
D_n(\mu(N+1,N))\left(T_N+\cdots\right)+\sum_{T_\beta\in P,\, m(\beta)\leq N-1}D_n(\mu(N+1,\beta))\left(T_\beta+\cdots\right).
\end{align*}
Сравнивая коэффициенты при операторе $T_N$ справа и слева приходим к равенству:
\begin{align}\label{eq-A0}
D_n(\mu(N+1,N))-\mu(N+1,N)=A_0-c\left(A_0u_n^0+A_3\right).
\end{align}
Поскольку в \eqref{eq-A0} функции $A_0$ и $A_3$ зависят только от разности $u_n^0-u_{n+1}^0$, то очевидно, что коэффициент $\mu=\mu(N+1,N)$ может зависеть только от $u_n^0$. Исследуем уравнение \eqref{eq-A0} на искомую функцию $\mu=\mu(u_n^0)$:
\begin{align}\label{maineq-A0}
\mu(u_{n+1}^0)-\mu(u_n^0)=A_0(u^0_n-u^0_{n+1})-c\left(A_0(u^0_n-u^0_{n+1})u_n^0+A_3(u^0_n-u^0_{n+1})\right).
\end{align}
Дифференцируя равенство \eqref{maineq-A0} сначала по $u_{n+1}^0$, а затем по $u_{n}^0$ получим дифференциальные уравнения:
\begin{enumerate}
	\item[i)] $\mu'(u_{n+1}^0)=-A_0'(u^0_n-u^0_{n+1})+c\left(A_0'(u^0_n-u^0_{n+1})u_n^0+A_3'(u^0_n-u^0_{n+1})\right)$; 
	\item[ii)] $0=-A_0''(u^0_n-u^0_{n+1})+cA_0''(u^0_n-u^0_{n+1})u_n^0+c\left(A_0'(u^0_n-u^0_{n+1})+A_3''(u^0_n-u^0_{n+1})\right)$.
\end{enumerate}
Отсюда следует, что должны выполняться равенства:
\begin{align*}
A_0''(u^0_n-u^0_{n+1})=0,\qquad
c\left(A_0'(u^0_n-u^0_{n+1})+A_3''(u^0_n-u^0_{n+1})\right)=0.
\end{align*}
Решая первое из уравнений находим
\begin{align*}
A_0(u^0_n-u^0_{n+1})=c_0(u^0_n-u^0_{n+1})+c_{0,0},
\end{align*}
где $c_0$ и $c_{0,0}$ некоторые константы. Если $c\neq 0$, то из второго уравнения находим функцию $A_3$:
\begin{align*}
A_3(u^0_n-u^0_{n+1})=-\frac{c_0}{2}(u^0_n-u^0_{n+1})^2+c_{3,1}(u^0_n-u^0_{n+1})+c_{3,0}.
\end{align*}
Найденные выражения для $A_0$ и $A_3$ подставим в уравнение $i)$ выше и получим условие 
\begin{align*}
\mu'(u_{n+1}^0)=-c_0+cc_0u_n^0+c\left(-c_0(u^0_n-u^0_{n+1})+c_{3,1}\right),
\end{align*}
из которого следует, что 
\begin{align*}
\mu(u_{n+1}^0)=\frac{cc_0}{2}(u_{n+1}^0)^2+\left(cc_{3,1}-c_0\right)u_{n+1}^0+k,
\end{align*}
где $k$ -- некоторая постоянная.
Подставляя найденные функции в равенство \eqref{maineq-A0} после несложных преобразований приходим к соотношению:
\begin{align*}
c_{0,0}-cc_{0,0}u_n^0-cc_{3,0}=0.
\end{align*}
Отсюда, в силу предположения выше о том, что $c\neq 0$ получим $c_{0,0}=0$ и $c_{3,0}=0$. Следовательно, искомые функции имеют вид 
\begin{align*}
A_0(u^0_n-u^0_{n+1})=c_0(u^0_n-u^0_{n+1}), \quad A_3(u^0_n-u^0_{n+1})=-\frac{c_0}{2}(u^0_n-u^0_{n+1})^2+c_{3,1}(u^0_n-u^0_{n+1}).
\end{align*}
Остается исследовать случай $c=0$. В этом случае \eqref{maineq-A0} принимает вид
\begin{align*}
\mu(u_{n+1}^0)-\mu(u_n^0)=A_0(u^0_n-u^0_{n+1}).
\end{align*}
Исследуем это уравнение следуя схеме, использованной выше и найдем $A_0(u^0_n-u^0_{n+1})=c_0(u^0_n-u^0_{n+1})$,
при этом никаких условий на функцию $A_3(u^0_n-u^0_{n+1})$ не получается. Окончательно имеем
\begin{align*}
A_0(u^0_n-u^0_{n+1})=c_0(u^0_n-u^0_{n+1}).
\end{align*}
Теорема \ref{ThA0} доказана.
\end{proof}

\subsection{Уточнение функций $A_1$, $A_2$ и $A_5$}

Накладывая на цепочку \eqref{jnx-part}, \eqref{form-f} вырожденные условия обрыва $u^{-1}_n=0$ и $u^{2}_n=0$ получаем систему двух уравнений
\begin{equation}\label{sys2x2}
\begin{cases}
u^0_{n+1,x}=u^0_{n,x}+A_2u^1_nu^0_n+A_4u^1_n+A_5u^0_n+A_7,\\
u^1_{n+1,x}=u^1_{n,x}+\bar{A}_1u^1_nu^{0}_{n+1}+\bar{A}_5u^1_n+\bar{A}_6u^{0}_{n+1}+\bar{A}_7,
\end{cases}
\end{equation}
где $A_i=A_i(u^0_n-u^0_{n+1})$, $\bar{A}_i=A_i(u^1_n-u^1_{n+1})$. 
При таких условиях обрыва характеристические операторы $X_0$, $X_1$ и $Y$ имеют вид
\begin{align*}
X_0=&\sum^{\infty}_{k=-\infty}\frac{\partial}{\partial u^0_k}, \quad X_1=\sum^{\infty}_{k=-\infty}\frac{\partial}{\partial u^1_k},\\
Y=&f\frac{\partial}{\partial u^0_1}+g\frac{\partial}{\partial u^1_1}+(f+f_1)\frac{\partial}{\partial u^0_2}+(g+g_1)\frac{\partial}{\partial u^1_2}-\\
&-f_{-1}\frac{\partial}{\partial u^0_{-1}}+g_{-1}\frac{\partial}{\partial u^1_{-1}}-(f_{-1}+f_{-2})\frac{\partial}{\partial u^0_{-2}}-(g_{-1}+g_{-2})\frac{\partial}{\partial u^1_{-2}}+\cdots,
\end{align*}
где $f=A_2u^1_nu^0_n+A_4u^1_n+A_5u^0_n+A_7$, $g=\bar{A}_1u^1_nu^{0}_{n+1}+\bar{A}_5u^1_n+\bar{A}_6u^{0}_{n+1}+\bar{A}_7$, $f_i=D_n^if$, $g_i=D_n^ig$. 

Ниже нам понадобятся операторы 
\begin{align*}
&Y_0=\left[X_0,Y\right]=\left(A_2u^1_n+A_5\right)\frac{\partial}{\partial u^0_1}+\left(\bar{A}_1u^1_n+\bar{A}_6\right)\frac{\partial}{\partial u^1_1}+\cdots,\\
&Y_{0,1}=\left[X_1,Y_0\right]=A_2u^1_n\frac{\partial}{\partial u^0_1}+\bar{A}_1\frac{\partial}{\partial u^1_1}+\cdots.
\end{align*}
Пользуясь этими операторами построим последовательность операторов, лежащих в характеристической алгебре системы \eqref{sys2x2}:
\begin{align}\label{seqA1A2}
T=Y_0,\quad T_0=Y_{0,1}, \quad T_1=\left[T,T_0\right], \quad T_2=\left[T,T_1\right], \quad \ldots.
\end{align}
На элементы этой последовательности автоморфизм характеристической алгебры действует следующим образом
\begin{align*}
&D_nX_0D^{-1}_n=X_0, \quad D_nX_1D^{-1}_n=X_1, \quad D_nYD^{-1}_n=Y-fX_0-gX_1,\\
&D_nTD^{-1}_n=T-\left(A_2u^1_n+A_5\right)X_0-\left(\bar{A}_1u^1_n+\bar{A}_6\right)X_1,\\
&D_nT_0D^{-1}_n=T_0-A_2X_0-\bar{A}_1X_1,\\
&D_nT_1D^{-1}_n=T_1-\bar{A}_1T_0+b_{1,1}X_0+b_{2,1}X_1,
\end{align*}
где 
$b_{1,1}=-T(A_2)+T_0\left(A_2u^1_n+A_5\right)-\bar{A}_1A_2, \quad b_{2,1}=-T(\bar{A}_1)+T_0\left(\bar{A}_1u^1_n+\bar{A}_6\right)-(\bar{A}_1)^2.$
\begin{align*}
D_nT_2D^{-1}_n=&T_2-\bar{A}_1T_1+T(\bar{A}_1)T_0-b_{2,1}T_0-\\
&-\left(A_2u^1_n+A_5\right)\left[X_0,T_1\right]-\left(\bar{A}_1u^1_n+\bar{A}_6\right)\left[X_1,T_1\right]+b_{1,2}X_0+b_{2,2}X_1.
\end{align*}
Можно показать, что операторы последовательности коммутируют с $X_0$, т.е. $\left[X_0,T_k\right]=0$ при $k\geq1$. А коммутирование $T_k$ с $X_1$ приводят к операторам, вообще говоря, не выражающимся через элементы последовательности \eqref{seqA1A2}. Поэтому, для того, чтобы получить замкнутую относительно действия автоморфизма последовательность, необходимо дополнить \eqref{seqA1A2} операторами, пронумерованными мультииндексами. По аналогии с предыдущим разделом определим мультииндекс $\alpha=(k,0,i_1,i_2,\ldots,i_{r-1},i_r)$ по правилу 
\begin{equation}\label{Talpha-A1A2}
T_\alpha= 
\begin{cases}
\left[X_1,T_{k,0,i_1,i_2,\ldots,i_{r-1}}\right], &\mbox{если}\quad i_r=0,\\
\left[T,T_{k,0,i_1,i_2,\ldots,i_{r-1}}\right], &\mbox{если}\quad i_r=1.
 \end{cases}
\end{equation}
Опишем действие автоморфизма на операторы $T_\alpha$:
\begin{equation*}
D_nT_\alpha D^{-1}_n=T_\alpha+\sum_{m(\beta)\leq m(\alpha)-1}\eta(\alpha,\beta)T_{\beta}.
\end{equation*}
Функция $m(\beta)$ определена в \eqref{malpha-A0}. Для элементов последовательности \eqref{seqA1A2} действие автоморфизма определяется равенством
\begin{align*}
D_nT_kD^{-1}_n=&T_k-\bar{A}_1T_{k-1}-\left(\bar{A}_1u^1_n+\bar{A}_6\right)\sum_{m(\beta)=k-1}T_{\beta}+\sum_{m(\beta)\leq k-2}\eta(k,\beta)T_{\beta}.
\end{align*}
\begin{theorem}\label{ThA1A2}
Если характеристическая алгебра системы \eqref{sys2x2} имеет конечную размерность, то 
$$A_1(\tau)=c_1\tau, \qquad A_2(\tau)=c_2\tau, \qquad A_5(\tau)=c_5\tau, \qquad A_7(\tau)=-\frac{c_5}{2}\tau^2+c_7\tau,$$
где $c_1$, $c_2$, $c_5$, $c_7$ -- некоторые постоянные.
\end{theorem}
Доказательство теоремы \ref{ThA1A2} вполне аналогично доказательству теоремы \ref{ThA0}, поэтому мы дадим его схематично. При доказательстве того, что $A_1(\tau)=c_1\tau$ мы используем расширенную последовательность, полученную объединением последовательностей \eqref{seqA1A2} и \eqref{Talpha-A1A2}. Для доказательства равенства $A_2(\tau)=c_2\tau$ мы строим аналогичные последовательности, беря в качестве $T$ и $T_0$ операторы $Y_1=\left[X_1,Y\right]$ и $Y_{0,1}=\left[X_0,Y_1\right]$.

Для уточнения вида функции $A_5$ воспользуемся редукцией цепочки \eqref{jnx-part}, \eqref{form-f}: 
\begin{equation}\label{1-eq}
u_{n+1,x}=u_{n,x}+A_5u_n+A_7,
\end{equation}
полученной наложением вырожденных граничных условий $u^{-1}_n=0$ и $u^{1}_n=0$. 
Класс уравнений вида $u_{n+1,x}=u_{n,x}+f(u_n,u_{n+1})$, являющихся интегрируемыми по Дарбу описан в работах \cite{HPZh2008}, \cite{HPZh2009}. Из результатов этой работы вытекает, что уравнение \eqref{1-eq} должно иметь вид
\begin{equation*}
u_{n+1,x}=u_{n,x}+\frac{c_5}{2}\left(u^2_n-u^2_{n+1}\right)+c_7\left(u_n-u_{n+1}\right), 
\end{equation*}
где $c_5$, $c_7$ -- постоянные параметры. Следовательно, 
\begin{align*}
A_5=c_5\left(u_n-u_{n+1}\right),\qquad
A_7=-\frac{c_5}{2}\left(u_n-u_{n+1}\right)^2+c_7\left(u_n-u_{n+1}\right).
\end{align*}

\subsection{Замена переменных в характеристических векторных полях}\label{3.3}

Для упрощения структуры характеристических операторов перейдем от стандартного набора динамических переменных $\left\{  u^j_k\right\}$ к новым переменным  $\left\{ \tau ^j_k, \bar u^j_0\right\}$, полагая $\tau ^j_k=u^j_k-u^j_{k+1}$, $\bar u^j_0= u^j_0$. Тогда операторы дифференцирования преобразуются по правилу
$$\frac{\partial}{\partial u^j_0}=\frac{\partial}{\partial \bar u^j_0}+\frac{\partial}{\partial \tau^j_0}-\frac{\partial}{\partial \tau^j_{-1}},$$
а также при $k\neq 0$
$$\frac{\partial}{\partial u^j_k}=\frac{\partial}{\partial \tau^j_k}-\frac{\partial}{\partial \tau^j_{k-1}}.$$
В результате характеристические операторы принимают более простой компактный вид
\begin{align*}
Y=\sum_{j=-N_2}^{N_1}\sum_{n=-\infty}^{n=+\infty}f^j_{n}\frac{\partial}{\partial \tau^j_{n}}, \qquad X_j=\frac{\partial}{\partial u^j_{0}}, \quad j=-N_2,-N_2+1,\ldots,N_1.
\end{align*}
Здесь и всюду ниже мы опускаем черту над $u^j_0$.

Положим $Y_j:=[X_j,Y]$ и сосредоточимся на подалгебре $L'_x$ характеристической алгебры, порожденной операторами  $\left\{ Y_j\right\}^{N_1}_{j=-N_2}$. Оператор  $Y_j$ имеет следующий явный вид
\begin{align}
Y_j=\sum_{k=-\infty}^{+\infty}X_j(f^{j-1}_{k})\frac{\partial}{\partial \tau^{j-1}_{k}}+
X_j(f^{j}_{k})\frac{\partial}{\partial \tau^{j}_{k}}+
X_j(f^{j+1}_{k})\frac{\partial}{\partial \tau^{j+1}_{k}}. \label{-1}
\end{align}
Предполагается, что в представлении \eqref{-1} функции $f^j_k=f(u^{j+1}_{k},u^{j}_k,u^j_{k+1 },u^{j-1}_{k+1})$ переписаны в новых переменных с учетом равенств
$$u^j_k=u^j_0-\sum_{m=0}^{k-1}\tau^{j}_{m}\quad \mbox{для} \quad k>0,\qquad \quad u^j_k=u^j_0+\sum_{m=-1}^{k}\tau^{j}_{m}\quad \mbox{для} \quad k<0.$$

Формулу \eqref{-1} приведем к следующему виду, собрав коэффициенты при независимых переменных  $u^{j-1}_0$, $u^j_0$, $u^{j+1}_0$ и т.д.:
\begin{align*}
Y=&u^{j-2}_0u^{j-1}_0S_0^{j-1}+u^{j-1}_0u^{j+1}_0S_0^{j+1}+u^{j+1}_0u^{j+2}_0S_0^{j+1}+ \\
&+u^{j-2}_0S_0^{j-1}+u^{j+2}_0S_1^{j+1}+u^{j-1}_0S_2^{j}+u^{j+1}_0S_2^{j+1}+S_3^j,
\end{align*}
где $S^j_m$ -- преобразованные операторы, имеющие громоздкие выражения. После некоторого
упрощения, связанного с заменой $\tau^j_k=e^{w^j_k}$, они приводятся к виду
\begin{align}
S^j_0=&\sum_k c_0 \frac{\partial}{\partial w^{j}_{k}},\qquad \quad 
S^j_1=\sum_k c_0 \tilde \rho^j_k\frac{\partial}{\partial w^{j}_{k}}+\tilde A_3(w^{j}_{k})\frac{\partial}{\partial w^{j}_{k}},\nonumber \\
S^j_2=&\sum_k c_2\frac{\partial}{\partial w^{j-1}_{k}} +c_1\frac{\partial}{\partial w^{j}_{k}}+ 
c_0\tilde \rho^{j-2}_{k+1}\frac{\partial}{\partial w^{j-1}_{k}}+
  c_0 \tilde \rho^{j+1}_k\frac{\partial}{\partial w^{j}_{k}}, \nonumber\\
S_3^j=&\sum_k c_5 \frac{\partial}{\partial w^{j}_{k}}	+c_2 \tilde \rho^{j-1}_k\frac{\partial}{\partial w^{j-1}_{k}}+
\tilde A_4(w^{j-1}_{k})\frac{\partial}{\partial w^{j-1}_{k}} + c_0 \tilde \rho^{j-1}_k \tilde \rho^{j-2}_{k+1}\frac{\partial}{\partial w^{j-1}_{k}}+ \label{-2}\\
&+\tilde A_3(w^{j-1}_{k})\tilde \rho^{j-2}_{k+1}\frac{\partial}{\partial w^{j-1}_{k}}+ c_0 \tilde \rho^{j+1}_k \tilde \rho^{j-1}_{k+1}\frac{\partial}{\partial w^{j}_{k}}  + c_1 \tilde \rho^{j-1}_{k+1}\frac{\partial}{\partial w^{j}_{k}}+
c_2 \tilde \rho^{j+1}_k \frac{\partial}{\partial w^{j}_{k}}+ \nonumber\\
&+\tilde A_6(w^{j+1}_{k})\frac{\partial}{\partial w^{j+1}_{k}}+c_0 \tilde \rho^{j+1}_k \tilde \rho^{j+2}_k \frac{\partial}{\partial w^{j+1}_{k}}+c_1\tilde \rho^{j+1}_k \frac{\partial}{\partial w^{j+1}_{k}}+ \tilde A_3(w^{j+1}_{k})\tilde \rho^{j+2}_k \frac{\partial}{\partial w^{j+1}_{k}}, \nonumber
\end{align}
где $\tilde A_s(w^j_k)=A_s(e^{w^j_k})e^{-w^j_k}$ для $s=3,4,6$, функция $\tilde \rho^{j}_i$ имеет представление
\begin{equation}\label{Trho}
\tilde \rho^{j}_i= 
\begin{cases}
-e^{w^j_0}-e^{w^j_1}-e^{w^j_2}-\cdots-e^{w^j_{i-1}}, & i\geq1,\\
0, & i=0, \\
e^{w^j_{-1}}+e^{w^j_{-2}}+e^{w^j_{-3}}+\cdots+e^{w^j_{i}}, & i\leq -1.
 \end{cases}
-N_2+1\leq j \leq N_1-1,
\end{equation}

Воспользовавшись условием конечномерности построенной выше подалгебры $L'_x$ характеристической алгебры можно доказать следующее утверждение, уточняющее вид искомой функции $f$.

\begin{theorem}
Если характеристическая алгебра $L_x$ имеет конечную размерность, то $c_0=0$.
\end{theorem}

\begin{proof}
Предположим противное, т.е. пусть $c_0\neq 0$. Рассмотрим операторы 
\begin{align*}
W_0=\left[S_0^{j+2},S_2^{j+1}\right]=c_0^2\sum_{k}\tilde \rho^{j+2}_k\frac{\partial}{\partial w^{j+1}_{k}},\qquad
W_1=\left[S_0^{j+1},S_2^{j}\right]=c_0^2\sum_{k}\tilde \rho^{j+1}_k\frac{\partial}{\partial w^{j}_{k}}.
\end{align*}
Построим последовательность кратных коммутаторов этих операторов:
\begin{align*}
&W_2=\left[W_0,W_1\right]=c_0^4\sum_{k}\tilde\rho^{j+2}_k \tilde\rho^{j+1}_k\frac{\partial}{\partial w^{j}_{k}},\\
&W_3=\left[W_0,W_2\right]=c_0^6\sum_{k}\left(\tilde\rho^{j+2}_k\right)^2 \tilde\rho^{j+1}_k\frac{\partial}{\partial w^{j}_{k}},\\
&\ldots\ldots\ldots\ldots,\\
&W_N=\left[W_0,W_{N-1}\right]=c_0^{2N}\sum_{k}\left(\tilde\rho^{j+2}_k\right)^{N-1} \tilde\rho^{j+1}_k\frac{\partial}{\partial w^{j}_{k}}.
\end{align*}
Легко проверить, что линейная оболочка натянутая на элементы этой последовательности имеет бесконечную размерность. Действительно, для набора операторов $W_1$, $W_2$, \ldots, $W_N$ определитель части матрицы коэффициентов 
\begin{align*}
\begin{pmatrix}
c_0^2\tilde\rho^{j+1}_0 & c_0^2\tilde\rho^{j+1}_1 & \ldots & c_0^2\tilde\rho^{j+1}_{N-1}\\
c_0^4\tilde\rho^{j+1}_0\tilde\rho^{j+2}_0 & c_0^4\tilde\rho^{j+1}_1\tilde\rho^{j+2}_1 & \ldots & c_0^4\tilde\rho^{j+1}_{N-1}\tilde\rho^{j+2}_{N-1}\\
\ldots & \ldots & \ldots & \ldots\\
c_0^{2N}\tilde\rho^{j+1}_0\left(\tilde\rho^{j+2}_0\right)^{N-1} & c_0^{2N}\tilde\rho^{j+1}_1\left(\tilde\rho^{j+2}_1\right)^{N-1} & \ldots & c_0^{2N}\tilde\rho^{j+1}_{N-1}\left(\tilde\rho^{j+2}_{N-1}\right)^{N-1}
\end{pmatrix}
\end{align*}
отличен от нуля при $c_0\neq 0$, поскольку он сводится к определителю Вандермонда. Поэтому, наше предположение, что $c_0\neq 0$ неверно. Теорема доказана.
\end{proof}
Условие $c_0=0$ заметно упрощает набор операторов \eqref{-2}:
\begin{equation*}
\begin{aligned}
S^j_1=&\sum_k \tilde A_3(w^{j}_{k})\frac{\partial}{\partial w^{j}_{k}},\qquad  \quad
S^j_2=\sum_k c_2\frac{\partial}{\partial w^{j-1}_{k}} +c_1\frac{\partial}{\partial w^{j}_{k}},\\
S_3^j=&\sum_k c_5 \frac{\partial}{\partial w^{j}_{k}}	+c_2 \tilde \rho^{j-1}_k\frac{\partial}{\partial w^{j-1}_{k}}+
\tilde A_4(w^{j-1}_{k})\frac{\partial}{\partial w^{j-1}_{k}}+\tilde A_3(w^{j-1}_{k})\tilde \rho^{j-2}_{k+1}\frac{\partial}{\partial w^{j-1}_{k}}+ c_1 \tilde \rho^{j-1}_{k+1}\frac{\partial}{\partial w^{j}_{k}}+\\
&+c_2 \tilde \rho^{j+1}_k \frac{\partial}{\partial w^{j}_{k}}+\tilde A_6(w^{j+1}_{k})\frac{\partial}{\partial w^{j+1}_{k}}+c_1\tilde \rho^{j+1}_k \frac{\partial}{\partial w^{j+1}_{k}}+ \tilde A_3(w^{j+1}_{k})\tilde \rho^{j+2}_k \frac{\partial}{\partial w^{j+1}_{k}}.
\end{aligned}
\end{equation*}
Основная цель наших построений ниже состоит в отыскании наиболее простых операторов в $L_x$ с тем, чтобы вывести эффективные условия на искомую функцию $f$ из конечномерности алгебры. 

Оператор $R^j_0=c_2S^j_2-c_1S^{j+1}_2\in L_x$ очевидно имеет следующий простой вид
\begin{equation}\label{-3}
R^j_0=c_2^2\sum_{k}\frac{\partial}{\partial w^{j-1}_{k}}-c_1^2\sum_{k}\frac{\partial}{\partial w^{j+1}_{k}}.
\end{equation}
Ниже нам понадобятся операторы $R^j_1:=\left[R^j_0,S^{j-1}_3\right]$ и $R^j_2:=\left[R^j_0,S^{j+1}_3\right]$, которые также имеют простые явные выражения
\begin{align*}
R^j_1=c_1^2\sum_{k}\tilde A_3(w^{j}_{k})\tilde \rho^{j+1}_k \frac{\partial}{\partial w^{j}_{k}},\qquad 
R^j_2=c_2^2\sum_{k}\tilde A_3(w^{j}_{k})\tilde \rho^{j-1}_{k+1} \frac{\partial}{\partial w^{j}_{k}}.
\end{align*}
Предположим, что выполняются неравенства $c_1\neq0$, $c_2\neq0$ и построим новый оператор $\bar{S}^j_3=S^j_3-\frac{1}{c_1^2}R^{j+1}_1-\frac{1}{c_2^2}R^{j-1}_2$. Дадим его явный вид:
\begin{equation*}
\begin{aligned}
\bar{S}_3^j=&\sum_k c_5 \frac{\partial}{\partial w^{j}_{k}}	+c_2 \tilde \rho^{j-1}_k\frac{\partial}{\partial w^{j-1}_{k}}+
\tilde A_4(w^{j-1}_{k})\frac{\partial}{\partial w^{j-1}_{k}}+ c_1 \tilde \rho^{j-1}_{k+1}\frac{\partial}{\partial w^{j}_{k}}+\\
&+c_2 \tilde \rho^{j+1}_k \frac{\partial}{\partial w^{j}_{k}}+\tilde A_6(w^{j+1}_{k})\frac{\partial}{\partial w^{j+1}_{k}}+c_1\tilde \rho^{j+1}_k \frac{\partial}{\partial w^{j+1}_{k}}.
\end{aligned}
\end{equation*}
Оператор $\bar{S}_3^j$ нуждается в дальнейшем упрощении. Рассмотрим коммутаторы $\bar{W}^j=\left[S^j_2,\bar{S}^{j+1}_3\right]$, $\bar{H}^j=\left[S^j_2,\bar{S}^{j-1}_3\right]$, которые представимы в виде
\begin{equation}\label{WH}
\begin{aligned}
&\bar{W}^j=\sum_k c_2 \tilde \rho^{j}_k\frac{\partial}{\partial w^{j}_{k}}+\tilde A_4'(w^{j}_{k})\frac{\partial}{\partial w^{j}_{k}}+ c_1 \tilde \rho^{j}_{k+1}\frac{\partial}{\partial w^{j+1}_{k}},\\
&\bar{H}^j=\sum_k c_2 \tilde \rho^{j}_k \frac{\partial}{\partial w^{j-1}_{k}}+\tilde A_6'(w^{j}_{k})\frac{\partial}{\partial w^{j}_{k}}+c_1\tilde \rho^{j}_k \frac{\partial}{\partial w^{j}_{k}}.
\end{aligned}
\end{equation}
Итогом наших построений являются следующие два оператора из $L_x$:
\begin{align*}
&\bar{Q}^j=\frac{1}{c_1}\left[S^j_2,\bar{W}^j\right]-\bar{W}^j=\sum_k\left(\tilde A_4''(w^{j}_{k})-\tilde A_4'(w^{j}_{k})\right)\frac{\partial}{\partial w^{j}_{k}},\\
&\tilde{Q}^j=\frac{1}{c_1}\left[S^j_2,\bar{H}^j\right]-\bar{H}^j=\sum_k\left(\tilde A_6''(w^{j}_{k})-\tilde A_6'(w^{j}_{k})\right)\frac{\partial}{\partial w^{j}_{k}}.
\end{align*}
Оба эти оператора имеют вид $Q^j=\sum_k q(w^{j}_{k})\frac{\partial}{\partial w^{j}_{k}}$. Сначала мы найдем возможные варианты функции $q(w)$ при условии конечномерности алгебры $L_x$, затем по найденным $q(w)$ определим возможные значения искомых функций $\tilde A_4(w)$ и $\tilde A_6(w)$, воспользовавшись равенствами $q(w)=\tilde A_4''(w)-\tilde A_4'(w)$ и $q(w)=\tilde A_6''(w)-\tilde A_6'(w)$.
\begin{theorem}\label{SQ}
Если алгебра Ли-Райнхарта, порожденная операторами 
$$S^j_2=\sum_k c_2\frac{\partial}{\partial w^{j-1}_{k}} +c_1\frac{\partial}{\partial w^{j}_{k}} \quad \mbox{и} \quad Q^j=\sum_k q(w^{j}_{k})\frac{\partial}{\partial w^{j}_{k}}$$ 
имеет конечную размерность, то функция $q(w)$ представима в одной из следующих форм:
\begin{enumerate}
	\item[1)] $q(w)=k_{0,1}+k_{1,1}w+k_{2,1}w^2$;
	\item[2)] $q(w)=k_{0,2}+k_{1,2}e^{\lambda w}+k_{2,2}e^{-\lambda w}$,
\end{enumerate}
где $k_{i,j}$ -- некоторые постоянные.
\end{theorem}
Доказательство теоремы можно найти в \cite{HPZh2008}.
\begin{theorem}\label{SQH-1}
Пусть $q(w)=k_{0,1}+k_{1,1}w+k_{2,1}w^2$ и алгебра Ли-Райнхарта, порожденная операторами 
\begin{align*}
&S^j_2=\sum_k c_2\frac{\partial}{\partial w^{j-1}_{k}} +c_1\frac{\partial}{\partial w^{j}_{k}}, \qquad Q^j=\sum_k q(w^{j}_{k})\frac{\partial}{\partial w^{j}_{k}},\\
&\bar{H}^j=\sum_k c_2 \tilde \rho^{j}_k \frac{\partial}{\partial w^{j-1}_{k}}+\tilde A_6'(w^{j}_{k})\frac{\partial}{\partial w^{j}_{k}}+c_1\tilde \rho^{j}_k \frac{\partial}{\partial w^{j}_{k}} 
\end{align*}
имеет конечную размерность. Тогда $k_{1,1}=0$ и $k_{2,1}=0$.
\end{theorem}
\begin{proof}
Оператор $Q^j$ имеет вид 
\begin{align*}
Q^j=\sum_k \left(k_{0,1}+k_{1,1}w^{j}_{k}+k_{2,1}(w^{j}_{k})^2\right)\frac{\partial}{\partial w^{j}_{k}}.
\end{align*}
Рассмотрим операторы следующего вида 
\begin{align*}
\left[S^j_2,Q^j\right]=\sum_k \left(k_{1,1}+2k_{2,1}w^{j}_{k}\right)\frac{\partial}{\partial w^{j}_{k}}, \qquad 
\left[S^j_2,\left[S^j_2,Q^j\right]\right]=2k_{2,1}\sum_k\frac{\partial}{\partial w^{j}_{k}}.
\end{align*}
Поскольку эти операторы принадлежат алгебре $L_x$, то при выполнении условия $k_{2,1}\neq 0$ операторы 
$$\sum_k\frac{\partial}{\partial w^{j}_{k}}, \qquad \sum_k w^{j}_{k}\frac{\partial}{\partial w^{j}_{k}}, \qquad \sum_k (w^{j}_{k})^2\frac{\partial}{\partial w^{j}_{k}}$$
также лежат в $L_x$.
Рассмотрим последовательность операторов 
\begin{align*}
&P^j_0=\sum_k w^{j}_{k}\frac{\partial}{\partial w^{j}_{k}},\\
&P^j_1=\left[P^j_0,\bar{H}^j\right]=\sum_k c_2w^{j}_{k}\tilde{\rho}^{j}_{k}\frac{\partial}{\partial w^{j-1}_{k}}+\sum_k F_1(w^{j}_{k})\frac{\partial}{\partial w^{j}_{k}},\\
&P^j_2=\left[P^j_0,P^j_1\right]=\sum_k c_2\left(w^{j}_{k}\right)^2\tilde{\rho}^{j}_{k}\frac{\partial}{\partial w^{j-1}_{k}}+\sum_k F_2(w^{j}_{k})\frac{\partial}{\partial w^{j}_{k}},
\end{align*}
где $F_1(w^{j}_{k})$ и $F_2(w^{j}_{k})$ -- некоторые функции, явный вид которых мы не уточняем. Нетрудно проверить, что в общем случае оператор $P^j_m=\left[P^j_0,P^j_{m-1}\right]$ имеет вид 
\begin{align*}
P^j_m=\sum_k c_2\left(w^{j}_{k}\right)^m\tilde{\rho}^{j}_{k}\frac{\partial}{\partial w^{j-1}_{k}}+\sum_k F_m(w^{j}_{k})\frac{\partial}{\partial w^{j}_{k}}, \quad m\geq 1.
\end{align*}
Ясно, что линейная оболочка векторных полей $\sum_k \left(w^{j}_{k}\right)^m\tilde{\rho}^{j}_{k}\frac{\partial}{\partial w^{j-1}_{k}}$ являющихся проекциями операторов $P^j_m$ образует при $1\leq m <\infty$ бесконечномерное пространство. Следовательно, предположение о том, что $k_{2,1}\neq 0$ неверно. Поэтому  $k_{2,1}=0$, аналогично можно доказать, что $k_{1,1}=0$. В итоге получаем $q(w)=k_{0,1}$. Теорема доказана.
\end{proof}

Полагая $q(w)=\tilde A_4''(w)-\tilde A_4'(w)$ и воспользовавшись теоремами \ref{SQ} и \ref{SQH-1} приходим к дифференциальному уравнению на искомую функцию $\tilde A_4(w)$: $\tilde A_4''(w)-\tilde A_4'(w)=k_4$, решение которого имеет вид 
\begin{equation}\label{A4-1}
A_4(w)=k_{4,0}+k_{4,1}e^w-k_4w.
\end{equation}
Аналогично можно показать, что 
\begin{equation}\label{A6-1}
A_6(w)=k_{6,0}+k_{6,1}e^w-k_6w.
\end{equation}
Исследуем случай $2)$ из теоремы \ref{SQ}.

\begin{theorem}\label{SQH-2}
Предположим, что алгебра Ли-Райнхарта, порожденная операторами 
\begin{align*}
&S^j_2=\sum_k c_2\frac{\partial}{\partial w^{j-1}_{k}} +c_1\frac{\partial}{\partial w^{j}_{k}}, \qquad Q^j=\sum_k q(w^{j}_{k})\frac{\partial}{\partial w^{j}_{k}},\\
&\bar{H}^j=\sum_k c_2 \tilde \rho^{j}_k \frac{\partial}{\partial w^{j-1}_{k}}+\tilde A_6'(w^{j}_{k})\frac{\partial}{\partial w^{j}_{k}}+c_1\tilde \rho^{j}_k \frac{\partial}{\partial w^{j}_{k}} 
\end{align*}
имеет конечную размерность, а функция $q(w)$ имеет вид: $q(w)=k_{0,2}+k_{1,2}e^{\lambda w}+k_{2,2}e^{-\lambda w}$. Тогда $q(w)=k_{0,2}$ или $q(w)=k_{0,2}+k_{3}e^{-\frac{1}{m} w}$, где $m$ -- целое положительное число, $k_3$~-- некоторая постоянная, отличная от нуля.
\end{theorem}
\begin{proof}
Из очевидных равенств
\begin{align*}
&\left[S^j_2,Q^j\right]=\lambda k_{1,2}\sum_k e^{\lambda w^{j}_{k}}\frac{\partial}{\partial w^{j}_{k}}-\lambda k_{2,2}\sum_k e^{-\lambda w^{j}_{k}}\frac{\partial}{\partial w^{j}_{k}},\\
&\left[S^j_2,\left[S^j_2,Q^j\right]\right]=\lambda^2 k_{1,2}\sum_k e^{\lambda w^{j}_{k}}\frac{\partial}{\partial w^{j}_{k}}+\lambda^2 k_{2,2}\sum_k e^{-\lambda w^{j}_{k}}\frac{\partial}{\partial w^{j}_{k}}
\end{align*}
следует, что операторы 
\begin{align*}
Q_\lambda = k_{1,2}\sum_k e^{\lambda w^{j}_{k}}\frac{\partial}{\partial w^{j}_{k}} \quad \mbox{и} 
\quad Q_{-\lambda} = k_{2,2}\sum_k e^{-\lambda w^{j}_{k}}\frac{\partial}{\partial w^{j}_{k}}
\end{align*}
лежат в $L_x$. Рассмотрим последовательность операторов
$$Q_\lambda, \quad Q_1=\left[Q_\lambda,\bar{H}^j\right], \quad Q_2=\left[Q_\lambda,Q_1\right], \quad Q_3=\left[Q_\lambda,Q_2\right], \quad \ldots.$$
Для этих операторов легко можно найти явное представление:
\begin{align*}
&Q_1=k_{1,2}c_2\sum_k e^{\lambda w^{j}_{k}}\tilde{\rho}^j_k\frac{\partial}{\partial w^{j-1}_{k}}+\sum_k F_1(w^{j}_{k})\frac{\partial}{\partial w^{j}_{k}},\\
&Q_2=k^2_{1,2}c_2(\lambda+1)\sum_k e^{2\lambda w^{j}_{k}}\tilde{\rho}^j_k\frac{\partial}{\partial w^{j-1}_{k}}+\sum_k F_2(w^{j}_{k})\frac{\partial}{\partial w^{j}_{k}},\\
&\ldots\ldots\ldots,\\
&Q_m=k^m_{1,2}c_2(\lambda+1)(2\lambda+1)\ldots((m-1)\lambda+1)\sum_k e^{m\lambda w^{j}_{k}}\tilde{\rho}^j_k\frac{\partial}{\partial w^{j-1}_{k}}+\sum_k F_m(w^{j}_{k})\frac{\partial}{\partial w^{j}_{k}}.
\end{align*}
Для того, чтобы последовательность $Q_\lambda$, $Q_1$, $Q_2$, $\ldots$ имела конечную размерность необходимо, чтобы для некоторого целого $m_1>0$ выполнялось равенство $m_1\lambda+1=0$ или $k_{1,2}=0$. Заменяя $\lambda \mapsto -\lambda$ из этих же соображений находим, что для некоторого $m_2>0$ должно выполняться $m_2\lambda-1=0$ или $k_{2,2}=0$. В итоге получаем, что функция $q(w)$ может иметь одно из следующих представлений:
\begin{enumerate}
	\item[1)] $q(w)=k_{0,2}$,
	\item[2)] $q(w)=k_{0,2}+k_{1,2}e^{-\frac{1}{m_1}}w$,
	\item[3)] $q(w)=k_{0,2}+k_{2,2}e^{-\frac{1}{m_2}}w$,
\end{enumerate}
где $m_1$, $m_2$ -- целые положительные числа. Фактически должно выполняться одно из двух условий:
\begin{enumerate}
	\item[1)] $q(w)=k_{0,2}$,
	\item[2)] $q(w)=k_{0,2}+k_3e^{-\frac{1}{m}}w$,
\end{enumerate}
где $k_{0,2}$ и $k_3$ -- некоторые постоянные, а $m$~-- целое положительное число. Теорема \ref{SQH-2} доказана.
\end{proof}
Воспользуемся теоремой \ref{SQH-2} для уточнения функций $\tilde A_4(w)$ и $\tilde A_6(w)$. В силу сказанного выше имеем уравнения
\begin{align*}
\tilde A_4''(w)-\tilde A_4'(w)=k_{0,2}+k_3e^{-\frac{1}{m}}w, \qquad \tilde A_6''(w)-\tilde A_6'(w)=\bar{k}_{0,2}+\bar{k}_3e^{-\frac{1}{\bar{m}}}w.
\end{align*}
Их решения имеют вид:
\begin{align}
&\tilde A_4(w)=r_0+r_1e^w-k_{0,2}w+\frac{m^2}{m+1}k_3e^{-\frac{1}{m}}w, \label{A4-2} \\
&\tilde A_6(w)=\bar{r}_0+\bar{r}_1e^w-\bar{k}_{0,2}w+\frac{\bar{m}^2}{\bar{m}+1}\bar{k}_3e^{-\frac{1}{\bar{m}}}w. \label{A6-2}
\end{align}
В дальнейшем мы будем пользоваться представлениями \eqref{A4-2} и \eqref{A6-2} для функций $\tilde A_4(w)$ и $\tilde A_6(w)$, поскольку ранее найденные представления \eqref{A4-1} и \eqref{A6-1} являются частными случаями \eqref{A4-2} и \eqref{A6-2}.

\begin{theorem}
Если характеристическая алгебра $L_x$ имеет конечную размерность, то $c_1=0$ и $c_2=0$.
\end{theorem}
\begin{proof} Воспользуемся ранее введенным оператором $\bar W^j$, переписав его с учетом конкретного представления \eqref{A4-2} искомой функции $\tilde A_4(w)$:
$$\bar W^j=\sum_k \left(c_2 \tilde \rho^{j}_k+ r_1e^{w^{j}_{k}}- k_{0,2}-\frac{m}{m+1}
  k_3e^{-\frac{1}{m}w^{j}_{k}} \right) \frac{\partial}{\partial w^{j}_{k}} 
+c_1 \tilde \rho^{j}_{k+1} \frac{\partial}{\partial w^{j+1}_{k}}.$$
Вычислим коммутатор операторов $S^j_2\in L_x$ и $\bar Q^j\in L_x$, имея в виду уточненное представление оператора $\bar Q^j=\sum_k (k_{0,2}+
 k_3e^{-\frac{1}{m}w^{j}_{k}} ) \frac{\partial}{\partial w^{j}_{k}}$
$$[S^j_2, \bar Q^j]=-\sum_k \frac{c_1}{m}
 k_3e^{-\frac{1}{m}w^{j}_{k}}  \frac{\partial}{\partial w^{j}_{k}}\in L_x.$$
Но тогда очевидно имеем, что оператор $\sum_k k_{0,2} \frac{\partial}{\partial w^{j}_{k}}$, а также  оператор
\begin{align}\label{Wj}
W^j=\sum_k (c_2 \tilde \rho^{j}_k+r_1e^{w^{j}_{k}}) \frac{\partial}{\partial w^{j}_{k}} 
+c_1 \tilde \rho^{j}_{k+1} \frac{\partial}{\partial w^{j+1}_{k}}
\end{align}
тоже лежат в $L_x$.

Рассмотрим теперь последовательность операторов в $L_x$, заданную в виде
\begin{align}\label{seq0}
W^j,\quad W^{j-1}, \quad W^j_1=[W^{j-1},W^j], \quad W^j_2=[W^{j-1},W^j_1], \quad W^j_3=[W^{j-1},W^j_2], \quad \dots.
\end{align}
Выясним, как действует автоморфизм \eqref{auto-main} на элементы последовательности. Для этого нам понадобятся следующие простые соотношения
\begin{align*}
&D_nW^jD^{-1}_n=W^j+e^{w^{j}_{0}}S_2^{j+1}, \qquad [S_2^j, W^j]=c_1W^j,\\
&[S_2^j, W^{j-1}]=c_2W^{j-1}, \qquad
[S_2^j, W^j_m]=(mc_2+c_1)W^j_m \quad \mbox{при}\quad m\geq1,
\end{align*}
при помощи которых можно показать, что имеют место равенства
\begin{align*}
D_nW^j_1D^{-1}_n=&W^j_1+c_1e^{w^{j-1}_{0}}W^{j},\\
D_nW^j_2D^{-1}_n=&W^j_2+(2c_1+c_2)e^{w^{j-1}_{0}}W^{j}_1+c_1(r_1+c_1+c_2)e^{w^{j-1}_{0}}W^{j},\\
D_nW^j_3D^{-1}_n=&W^j_3+3(c_1+c_2)e^{w^{j-1}_{0}}W^{j}_2+(c_2(2c_1+c_2)+
(3c_1+c_2)(r_1+c_1+c_2))e^{2w^{j-1}_{0}}W^{j}_1\\
&+c_1(c_1+c_2+r_1)(c_1+2r_1+c_2)e^{3w^{j-1}_{0}}W^{j}.
\end{align*}
В общем случае имеем
\begin{align*}
D_nW^j_mD^{-1}_n=W^j_m+\left(mc_1+\frac{m(m-1)}{2}c_2\right)e^{w^{j-1}_{0}}W^{j}_{m-1}+\dots, \quad m\geq 2,
\end{align*}
где в многоточие отнесена линейная комбинация младших членов последовательности.

Поскольку  $L_x$ имеет конечную размерность, то существует такое целое число $M\geq 1$, что имеет место представление
\begin{align}\label{L1}
W^j_{M+1}=\lambda W^j_M+\dots.
\end{align}
Применим автоморфизм \eqref{auto-main} к обеим частям последнего равенства и упростим результат в силу формул выше. В итоге придем к соотношению
$$\lambda W^j_M+\left((M+1)c_1+\frac{M(M+1)}{2}c_2\right)e^{w^{j-1}_{0}}W^j_M=D_n(\lambda)W^j_M+\dots.$$
Сравнивая коэффициенты при операторе $W^j_M$ получим уравнение на $\lambda$
$$D_n(\lambda)-\lambda=\left((M+1)c_1+\frac{M(M+1)}{2}c_2\right)e^{w^{j-1}_{0}}.$$
Функция $\lambda$ может зависеть только от конечного числа динамических переменных, поэтому полученное равенство может выполняться лишь в том случае, когда правая часть уравнения равна нулю. Отсюда следует условие на параметры при $M\geq 1$
\begin{align}\label{R1}
c_1=-\frac{M}{2}c_2.
\end{align} 
Случай $M=0$, т.е. когда выполняется равенство 
\begin{align}\label{Wj1}
W^j_{1}=\lambda W^j+\nu S_2^{j+1},
\end{align}
следует рассмотреть отдельно. Под действием автоморфизма равенство \eqref{Wj1} преобразуется к виду
$$\lambda W^j+\nu S^{j+1}_2+c_1e^{w^{j-1}_{0}}W^j=D_n(\lambda)(W^j+e^{w^{j}_{0}}S^{j+1}_2)+D_n(\nu)   S^{j+1}_2.$$
Сравнивая коэффициенты при независимых операторах $W^j$ и $S^{j+1}_2$ получаем уравнения
$$D_n(\lambda)-\lambda=c_1e^{w^{j-1}_{0}}, \qquad D_n(\nu)-\nu=-D_n(\lambda)e^{w^{j}_{0}}$$
из которых следует, что 
\begin{align}\label{L2}
c_1=0, \quad \lambda=0, \quad \nu=\const.
\end{align}

Схему рассуждений, предложенную выше, можно применить еще раз, заменив оператор $W^j$ оператором
\begin{align}\label{Hj}
H^j=\sum_k c_2 \tilde \rho^{j+1}_k \frac{\partial}{\partial w^{j}_{k}} 
+\left(\bar{r}_1e^{w^{j+1}_{k}} +  c_1 \tilde \rho^{j+1}_k\right) \frac{\partial}{\partial w^{j+1}_{k}},
\end{align}
который также лежит в $L_x$, в чем нетрудно убедиться, используя те же соображения, что и в случае $W^j$, заменив $\bar W^j$ на $\bar H^j$  (см. формулу \eqref{WH}). В этом случае вместо \eqref{seq0} мы воспользуемся последовательностью следующего вида
\begin{align}\label{seq1}
H^{j+1},\quad H^{j}, \quad H^j_1=[H^{j+1},H^j], \quad H^j_2=[H^{j+1},H^j_1], \quad H^j_3=[H^{j+1},H^j_2], \quad \dots.
\end{align}
Опуская подробное изложение всей конструкции, мы приведем только выводы. Если последовательность обрывается на первом шаге, то получаем $c_2=0$; если на шаге с номером $N\geq2$, то приходим к равенству $$c_1=-\frac{2}{N}c_2.$$ Комбинируя эти условия с теми, что были получены ранее (см. \eqref{R1}, \eqref{L2}) сделаем вывод, что должна реализоваться одна из следующих возможностей: 
\begin{enumerate}
	\item[1)] $c_1=0, \quad c_2=0;$
	\item[2)] $N=1, \quad M=4,  \quad c_1=-2c_2;$
	\item[3)] $N=2, \quad M=2,  \quad c_1=-c_2;$
	\item[4)] $N=4, \quad M=1,  \quad c_1=-\frac{1}{2}c_2.$
\end{enumerate}
Если реализуется случай 1), то теорема доказана. Предположим, что имеет место случай 2).  

Поскольку $N=1$, то имеем представление 
\begin{align}\label{R3}
H^j_2=\lambda H^j_1+\mu H^j+\eta S_2^{j+1},
\end{align}
из которого, подействовав автоморфизмом \eqref{auto-main}, можно получить соотношение
\begin{align*}
&\lambda H^j_1+\mu H^j+\eta S_2^{j+1}+c_2(\bar{r}_1-c_2)e^{2w^{j+2}_0}H^j+c_2(\bar{r}_1-c_2)e^{2w^{j+2}_0+w^{j+1}_0}S_2^{j+1}=\\
&=D_n(\lambda)\left(H_1^j+c_2e^{w^{j+2}_0}H^j+c_2e^{w^{j+2}_0+w^{j+1}_0}S_2^{j+1}\right)+D_n(\mu)\left(H^j+e^{w^{j+1}_0}S_2^{j+1}\right)+D_n(\eta)S_2^{j+1}.
\end{align*}
Сравним коэффициенты при независимых операторах и получим уравнения:
\begin{enumerate}
	\item[1)] $D_n(\lambda)-\lambda=0$;
	\item[2)] $D_n(\mu)-\mu=c_2(\bar{r}_1-c_2)e^{2w^{j+2}_0}-D_n(\lambda) c_2e^{w^{j+2}_0}$;
	\item[3)] $D_n(\eta)-\eta=c_2(\bar{r}_1-c_2)e^{2w^{j+2}_0+w^{j+1}_0}-D_n(\lambda)c_2e^{w^{j+2}_0+w^{j+1}_0}-D_n(\mu)e^{w^{j+1}_0}$.
\end{enumerate}
Из 1) имеем $\lambda=\const$, из 2): $\lambda=0$, $\bar{r}_1=c_2$, в силу предположения, что $c_2\neq 0$. Из 3) получаем $\mu=0$, $\eta=\const$. С учетом найденных коэффициентов разложение \eqref{R3} принимает вид:
\begin{align*}
H^j_2=\eta S_2^{j+1}, \quad \eta=\const.
\end{align*}
Но это противоречит явным формулам, так как
\begin{align*}
H^j_2=\sum_kc_2^2\tilde \rho^{j+2}_k\tilde \rho^{j+1}_k\left(e^{w^{j+2}_k}-\tilde \rho^{j+2}_k\right) \frac{\partial}{\partial w^{j}_{k}},\qquad
S_2^{j+1}=c_2\sum_k\frac{\partial}{\partial w^{j}_{k}}+c_1\sum_k\frac{\partial}{\partial w^{j+1}_{k}},
\end{align*}
поэтому случай 2) не реализуется, если хотя бы одно из чисел $c_1$ и $c_2$ отлично от нуля.

Аналогично проверяется, что случай 4) также не реализуется. Остается проверить случай 3): $N=2$, $M=2$, $c_1=-c_2$. Поскольку мы  имеем равенство
\begin{align*}
D_nW^j_3D^{-1}_n=W^j_3-c_2(c_2+2r_1)e^{2w^{j-1}_{0}}W^{j}_1-2c_2r_1^2e^{3w^{j-1}_{0}}W^{j},
\end{align*}
то случаю 3) соответствует разложение 
\begin{align*}
W^j_3=\lambda W^j_2+\mu W^{j}_1+\eta W^{j}+\theta S_2^{j+1},
\end{align*}
из которого, в силу \eqref{auto-main}, получаем соотношение
\begin{align*}
&D_n(\lambda)\left(W_2^j-c_2e^{w_0^{j-1}}W^{j}_1-c_2r_1e^{2w^{j-1}_{0}}W^{j}\right)+D_n(\mu)\left(W^{j}_1-c_2e^{w^{j-1}_{0}}W^{j}\right)+\\
&+D_n(\eta)\left(W^j+e^{w^j_0}S_2^{j+1}\right)+D_n(\theta) S_2^{j+1}=\lambda W^j_2+\mu W^{j}_1+\eta W^{j}+\theta S_2^{j+1}-\\
&-c_2(c_2+2r_1)e^{2w^{j-1}_{0}}W^{j}_1-2c_2r_1^2e^{3w^{j-1}_{0}}W^{j}.
\end{align*}
Сравнивая коэффициенты при независимых операторах получим уравнения:
\begin{enumerate}
	\item[1)] $D_n(\lambda)-\lambda=0$;
	\item[2)] $D_n(\mu)-\mu=D_n(\lambda)c_2e^{w_0^{j-1}}-c_2(c_2+2r_1)e^{2w^{j-1}_{0}}$;
	\item[3)] $D_n(\eta)-\eta=D_n(\lambda)c_2r_1e^{2w^{j-1}_{0}}+D_n(\mu)c_2e^{w^{j-1}_{0}}-2c_2r_1^2e^{3w^{j-1}_{0}}$;
	\item[4)] $D_n(\theta)-\theta=-D_n(\eta)e^{w^j_0}$.
\end{enumerate}
Из этих уравнений находим:
\begin{enumerate}
	\item[1)] $\lambda=\const$;
	\item[2)] $\lambda=0$, $\mu=\const$, $c_2=-2r_1$;
	\item[3)] $\mu=0$,  $\eta=\const$, $r_1=0$;
	\item[4)] $\eta=0$, $\theta=\const$.
\end{enumerate}
Отсюда видно, что случай 3) реализуется только в том случае, когда $c_2=0$. Но тогда и $c_1=0$. Теорема доказана. 
\end{proof}

\subsection{Уточнение функций $A_3$, $A_4$ и $A_6$}

В предыдущих разделах было показано, что часть функций в \eqref{form-f} равна нулю, в результате чего $f$ принимает вид:
\begin{align}\label{form-f2}
f\left(u^1_n, u^0_n, u^{0}_{n+1}, u^{-1}_{n+1}\right)=&A_3(\tau)u^1_nu^{-1}_{n+1}+A_4(\tau)u^1_n+A_5(\tau)u^0_n+A_6(\tau)u^{-1}_{n+1}+A_7(\tau),
\end{align}
где $\tau=u^0_n-u^{0}_{n+1}$, причем $A_3(\tau)$, $A_4(\tau)$, $A_6(\tau)$ остаются неизвестными функциями, в то время как функции $A_5(\tau)=c_5\tau$ и $A_7(\tau)=-\frac{c_5}{2}\tau^2+c_7\tau$ уже определены.

При этом подалгебра $L'_x$ характеристической алгебры $L_x$, введенная в предыдущем разделе, порождается операторами:
\begin{align}\label{eqS1S3}
\begin{aligned}
S^j_1=&\sum_k \tilde A_3(w^{j}_{k})\frac{\partial}{\partial w^{j}_{k}}, \quad -N_2\leq j\leq N_1,\\
S_3^j=&\sum_k c_5 \frac{\partial}{\partial w^{j}_{k}}+\tilde A_4(w^{j-1}_{k})\frac{\partial}{\partial w^{j-1}_{k}}+\tilde A_3(w^{j-1}_{k})\tilde \rho^{j-2}_{k+1}\frac{\partial}{\partial w^{j-1}_{k}}+\\
&+\tilde A_6(w^{j+1}_{k})\frac{\partial}{\partial w^{j+1}_{k}}+ \tilde A_3(w^{j+1}_{k})\tilde \rho^{j+2}_k \frac{\partial}{\partial w^{j+1}_{k}}.
\end{aligned}
\end{align}
Рассмотрим последовательность, порожденную кратными коммутаторами операторов $S^j_3$ и $S^{j+1}_1$:
\begin{align}\label{seqR}
R_1=\left[S^j_3,S^{j+1}_1\right], \quad R_2=\left[S^j_3,R_1\right], \quad R_3=\left[S^j_3,R_2\right], \quad \ldots.
\end{align}
Легко проверить, что 
\begin{align*}
&R_1=\sum_kc_5\tilde A'_3(w^{j+1}_{k})\frac{\partial}{\partial w^{j+1}_{k}}, \\
&R_2=\sum_kc_5\tilde A''_3(w^{j+1}_{k})\frac{\partial}{\partial w^{j+1}_{k}}, \\
&\ldots\ldots,\\
&R_m=\sum_kc_5\tilde A^{(m)}_3(w^{j+1}_{k})\frac{\partial}{\partial w^{j+1}_{k}}.
\end{align*}
В силу конечномерности характеристической алгебры $L_x$, для некоторого $M$, оператор $R_{M+1}$ должен линейно выражаться через операторы $S_1^{j+1}$, $R_1$, $R_2$, \ldots, $R_M$, которые являются линейно независимыми:
\begin{align}\label{P1}
R_{M+1}=\lambda_MR_M+\lambda_{M-1}R_{M-1}+\cdots+\lambda_1R_1+\lambda_0S^{j+1}_1.
\end{align}
Отметим, что автоморфизм \eqref{auto-main} переводит операторы последовательности \eqref{seqR} в себя. Применим автоморфизм к \eqref{P1} и получим
\begin{align}\label{P2}
R_{M+1}=D_n(\lambda_M)R_M+D_n(\lambda_{M-1})R_{M-1}+\cdots+D_n(\lambda_1)R_1+D_n(\lambda_0)S^{j+1}_1.
\end{align}
Отсюда ясно, что коэффициенты $\lambda_j$ -- постоянны. Переходя к координатным представлениям операторов легко показать, что $\tilde A_3(w)$ является квазиполиномом, который можно уточнить, пользуясь Леммой 6 работы \cite{HPZh2008} и доказать, что функция $\tilde A_3(w)$ имеет одну из следующих двух форм:
\begin{align}
	&1) \quad \tilde A_3(w)=c_{3,0}+c_{3,1}w+c_{3,2}w^2, \label{F1}\\
	&2) \quad \tilde A_3(w)=c_{3,0}+c_{3,1}e^{\lambda w}+c_{3,2}e^{-\lambda w}, \label{F2}
\end{align}
где $c_{i,j}$ -- некоторые постоянные.

Рассмотрим сначала более подробно  случай $1)$. 
\begin{lemma}\label{lem3.1}
Если $\tilde A_3(w)$ имеет вид \eqref{F1}, то справедливо равенство $\tilde A_3(w)\equiv0$.
\end{lemma}
\begin{proof}
С учетом формулы \eqref{F1} оператор $S^j_1$ представляется в виде:
\begin{align*}
S^j_1=&\sum_k \left(c_{3,0}+c_{3,1}w^{j}_{k}+c_{3,2}(w^{j}_{k})^2\right)\frac{\partial}{\partial w^{j}_{k}}.
\end{align*}
Легко можно показать, взяв кратные коммутаторы операторов $S^j_3$ и $S^j_1$, что при выполнении условия $c_{3,2}\neq 0$, операторы $\sum_k\frac{\partial}{\partial w^{j}_{k}}$, $\sum_kw^{j}_{k}\frac{\partial}{\partial w^{j}_{k}}$ и $\sum_k(w^{j}_{k})^2\frac{\partial}{\partial w^{j}_{k}}$ лежат в $L_x$.

Рассмотрим последовательность операторов, заданных в виде
\begin{align*}
Q^{j+3}, \quad P^j, \quad P^j_1=\left[Q^{j+3}, P^j\right], \quad P^j_2=\left[Q^{j+3}, P^j_1\right], \quad \ldots,
\end{align*}
где
\begin{align*}
&Q^{j+3}=\sum_kw^{j+3}_{k}\frac{\partial}{\partial w^{j+3}_{k}},\\
&P^j=ad^2_{S^{j+1}_3}S^{j}_3-c_{3,2}\sum_k\frac{\partial}{\partial w^{j+2}_{k}}=\sum_k\tilde A''_6(w^{j+2}_{k})\frac{\partial}{\partial w^{j+2}_{k}}+c_{3,2}\sum_k\tilde \rho^{j+3}_k \frac{\partial}{\partial w^{j+2}_{k}}, \\
&P^j_1=c_{3,2}\sum_kw^{j+3}_{k}\tilde \rho^{j+3}_k \frac{\partial}{\partial w^{j+2}_{k}},\\
&P^j_2=P^j_1+c_{3,2}\sum_k(w^{j+3}_{k})^2\tilde \rho^{j+3}_k \frac{\partial}{\partial w^{j+2}_{k}},\\
&P^j_3=3P^j_2-2P^j_1+c_{3,2}\sum_k(w^{j+3}_{k})^3\tilde \rho^{j+3}_k \frac{\partial}{\partial w^{j+2}_{k}}.
\end{align*}
По индукции можно доказать, что оператор вида $c_{3,2}\sum_k(w^{j+3}_{k})^m\tilde \rho^{j+3}_k \frac{\partial}{\partial w^{j+2}_{k}}$ лежит в $L_x$. Если $c_{3,2}\neq 0$, то линейная оболочка всех таких операторов имеет бесконечную размерность, что противоречит условию конечномерности алгебры $L_x$. Следовательно, $c_{3,2}=0$. Аналогично можно проверить, что $c_{3,1}=0$.

Докажем методом от противного, что $c_{3,0}=0$. Пусть $c_{3,0}\neq 0$, тогда операторы
\begin{align}\label{++}
R^j=\left[S^{j+1}_1,S^{j-1}_3\right]=\sum_kc_{3,0}\rho^{j+1}_k \frac{\partial}{\partial w^{j}_{k}},\qquad
\bar{R}^j=\left[S^{j-1}_1,S^{j+1}_3\right]=\sum_kc_{3,0}\rho^{j-1}_{k+1} \frac{\partial}{\partial w^{j}_{k}}
\end{align}
принадлежат $L_x$. Рассмотрим последовательность их кратных коммутаторов:
\begin{align*}
&R^j_1=\left[R^{j-1},\bar{R}^j\right]=\sum_kc_{3,0}^2\rho^{j}_k\rho^{j-1}_{k+1} \frac{\partial}{\partial w^{j}_{k}}-\sum_kc_{3,0}^2\rho^{j}_k\rho^{j-1}_{k+1} \frac{\partial}{\partial w^{j-1}_{k}},\\
&R^j_2=\left[R^{j-1},R^j_1\right]=\sum_kc_{3,0}^3(\rho^{j}_k)^2\rho^{j-1}_{k+1} \frac{\partial}{\partial w^{j}_{k}}-2\sum_kc_{3,0}^3(\rho^{j}_k)^2\rho^{j-1}_{k+1} \frac{\partial}{\partial w^{j-1}_{k}},
\end{align*}
для $m\geq 3$ имеем 
\begin{align*}
R^j_m=\left[R^{j-1},R^j_{m-1}\right]=\sum_kc_{3,0}^{m+1}(\rho^{j}_k)^m\rho^{j-1}_{k+1} \frac{\partial}{\partial w^{j}_{k}}-m\sum_kc_{3,0}^{m+1}(\rho^{j}_k)^m\rho^{j-1}_{k+1} \frac{\partial}{\partial w^{j-1}_{k}}.
\end{align*}
Если $c_{3,0}\neq 0$, то как легко показать, пользуясь формулой Вандермонда, что эта последовательность бесконечномерна, поэтому $c_{3,0}=0$. Лемма доказана.
\end{proof}

Согласно Лемме \ref{lem3.1} подалгебра $L'_x$ порождена семейством операторов
\begin{align*}
S_3^j=\sum_k c_5 \frac{\partial}{\partial w^{j}_{k}}+\tilde A_4(w^{j-1}_{k})\frac{\partial}{\partial w^{j-1}_{k}}+\tilde A_6(w^{j+1}_{k})\frac{\partial}{\partial w^{j+1}_{k}}, \quad -N_2\leq j\leq N_1.
\end{align*}
Возьмем операторы $S^j_3$ и $S^{j+1}_3$ и рассмотрим последовательность их кратных коммутаторов:
\begin{align*}
S_3^j, \quad S^{j+1}_3, \quad R_1=\left[S_3^j,S^{j+1}_3\right], \quad R_2=\left[S_3^j,R_1\right], \quad \ldots.
\end{align*}
Они имеют следующее явное представление:
\begin{align*}
&R_1=\sum_k c_5 \tilde A'_4(w^{j}_{k})\frac{\partial}{\partial w^{j}_{k}}-\sum_k c_5\tilde A'_6(w^{j+1}_{k})\frac{\partial}{\partial w^{j+1}_{k}},\\
&R_2=\sum_k c^2_5 \tilde A''_4(w^{j}_{k})\frac{\partial}{\partial w^{j}_{k}}+\sum_k F_2(w^{j+1}_{k})\frac{\partial}{\partial w^{j+1}_{k}},\\
&\ldots\ldots\ldots,\\
&R_m=\sum_k c^m_5 \tilde A^{m}_4(w^{j}_{k})\frac{\partial}{\partial w^{j}_{k}}+\sum_k F_m(w^{j+1}_{k})\frac{\partial}{\partial w^{j+1}_{k}}.
\end{align*}
Отсюда можно вывести, что функция $\tilde A_4(w)$ является квазимногочленом, для определения структуры которого можно воспользоваться схемой, изложенной в \cite{HPZh2008} (см. Лемму 6), и показать, что $\tilde A_4(w)$ принадлежит одному из классов:
\begin{align}
	&1) \quad \tilde A_4(w)=c_{4,0}+c_{4,1}w+c_{4,2}w^2, \label{F41}\\
	&2) \quad \tilde A_4(w)=c_{4,0}+c_{4,1}e^{\lambda w}+c_{4,2}e^{-\lambda w}. \label{F42}
\end{align}
Аналогично, рассматривая последовательность операторов
\begin{align*}
S_3^{j+1}, \quad S^{j}_3, \quad \bar{R}_1=\left[S^{j+1}_3,S_3^j\right], \quad \bar{R}_2=\left[S^{j+1}_3,\bar{R}_1\right], \quad \ldots
\end{align*}
можно убедиться, что функция $\tilde A_6(w)$ имеет одну из форм
\begin{align}
	&1) \quad \tilde A_6(w)=c_{6,0}+c_{6,1}w+c_{6,2}w^2, \label{F61}\\
	&2) \quad \tilde A_6(w)=c_{6,0}+c_{6,1}e^{\lambda w}+c_{6,2}e^{-\lambda w}. \label{F62}
\end{align}
Наметим схему доказательства утверждения, что функции $\tilde A_4(w)$ и $\tilde A_6(w)$ либо одновременно полиномиальны, либо экспоненциальны. Предположим, что $\tilde A_4(w)$ имеет вид \eqref{F41}, а $\tilde A_6(w)$ -- вид \eqref{F62}. Воспользуемся операторами:
\begin{align*}
S_3^j=\sum_k c_5 \frac{\partial}{\partial w^{j}_{k}}+\tilde A_4(w^{j-1}_{k})\frac{\partial}{\partial w^{j-1}_{k}}+\tilde A_6(w^{j+1}_{k})\frac{\partial}{\partial w^{j+1}_{k}},
\end{align*}
с верхними индексами $j$ и $j+2$. Легко видеть, что 
\begin{align*}
Q=\left[S_3^{j+2},S_3^j\right]=\sum_k\left(\tilde A_4(w^{j+1}_{k})\tilde A'_6(w^{j+1}_{k})-\tilde A'_4(w^{j+1}_{k})\tilde A_6(w^{j+1}_{k})\right)\frac{\partial}{\partial w^{j+1}_{k}}.
\end{align*}
Пусть задана последовательность
\begin{align*}
Q, \quad Q_1=\left[S_3^{j+2},Q\right], \quad Q_2=\left[S_3^{j+2},Q_1\right], \quad \ldots.
\end{align*}
Можно проверить, что коэффициент при $\frac{\partial}{\partial w^{j+1}_{k}}$  в операторе $Q_m$ имеет вид
 $P_m(w^{j+1}_{k})e^{\lambda w^{j+1}_{k}}$, где $P_m(w^{j+1}_{k})$~-- полином, степень $r_m$ которого
неограниченно возрастает по $m$. Поэтому эти операторы являются линейно независимыми. Нетрудно показать, что в \eqref{F42} и \eqref{F62} параметр $\lambda$ имеет одно и то же значение.

Теперь рассмотрим случай 2).
\begin{lemma}
Если функция $\tilde A_3(w)$ имеет вид \eqref{F2}, то 
\begin{align}\label{F32}
\tilde A_3(w)=c_{3}e^{-\frac{1}{m} w},
\end{align}
где $c_{3}$ совпадает с одной из двух констант $c_{3,1}$ и $c_{3,2}$, $m>0$~-- целое число.
\end{lemma}
\begin{proof}
 Построим новые операторы 
\begin{align*}
&Q_\lambda=\frac{1}{2\lambda^2}\left(\lambda ad_{S_3^{j}}S_1^{j}+ad^2_{S_3^{j}}S_1^{j}\right)=\sum_kc_{3,1}e^{\lambda w^{j}_{k}}\frac{\partial}{\partial w^{j}_{k}},\\
&Q_{-\lambda}=-\frac{1}{2\lambda^2}\left(\lambda ad_{S_3^{j}}S_1^{j}-ad^2_{S_3^{j}}S_1^{j}\right)=\sum_kc_{3,2}e^{-\lambda w^{j}_{k}}\frac{\partial}{\partial w^{j}_{k}},
\end{align*}
где $S_1^{j}$ и $S_3^{j}$ имеют вид \eqref{eqS1S3}. Возьмем последовательность операторов вида:
\begin{align*}
&ad_{Q_\lambda}S_3^{j-2}=c_{3,1}\sum_k\tilde A_3(w^{j-1}_{k})\tilde{\rho}^j_ke^{\lambda w^{j}_{k}}\frac{\partial}{\partial w^{j-1}_{k}},\\
&ad^2_{Q_\lambda}S_3^{j-2}=c^2_{3,1}(\lambda+1)\sum_k\tilde A_3(w^{j-1}_{k})\tilde{\rho}^j_ke^{2\lambda w^{j}_{k}}\frac{\partial}{\partial w^{j-1}_{k}},\\
&ad^3_{Q_\lambda}S_3^{j-2}=c^3_{3,1}(\lambda+1)(2\lambda+1)\sum_k\tilde A_3(w^{j-1}_{k})\tilde{\rho}^j_ke^{3\lambda w^{j}_{k}}\frac{\partial}{\partial w^{j-1}_{k}},\\
&\ldots\ldots\ldots,\\
&ad^m_{Q_\lambda}S_3^{j-2}=c^m_{3,1}(\lambda+1)(2\lambda+1)\ldots((m-1)\lambda+1)\sum_k\tilde A_3(w^{j-1}_{k})\tilde{\rho}^j_ke^{m\lambda w^{j}_{k}}\frac{\partial}{\partial w^{j-1}_{k}}.
\end{align*}
Поскольку $\lambda\neq 0$, то эта последовательность может оборваться только за счет зануления сомножителя $c_{3,1}(m\lambda+1)$. Поэтому если $c_{3,1}\neq 0$, то $\lambda=-\frac{1}{m}$, где $m$~-- целое положительное число.

Аналогично, рассматривая кратные коммутаторы операторов $Q_{-\lambda}$ и $S_3^{j-2}$ мы приходим к условию $c_{3,2}(\bar{m}\lambda-1)=0$, из которого в силу $c_{3,2}\neq 0$ находим $\lambda=\frac{1}{\bar{m}}$, где $\bar{m}>0$~-- целое число. Сравнивая найденные условия приходим к трем возможным вариантам:
\begin{align*}
1) \quad c_{3,1}=0, \quad c_{3,2}=0; \qquad
2) \quad c_{3,1}=0, \quad \lambda=\frac{1}{\bar{m}}; \qquad
3) \quad c_{3,2}=0, \quad \lambda=-\frac{1}{m}.
\end{align*}
Из первого варианта следует, что $\tilde A_3(w)=c_{3,0}$. Этот случай был рассмотрен ранее (см. Лемму~\ref{lem3.1}), где было показано, что тогда $\tilde A_3(w)=0$. Объединим случаи 2) и 3), представив $\tilde A_3(w)$ в виде 
\begin{align*}
\tilde A_3(w)=c_{3,0}+c_{3}e^{-\frac{1}{m} w},
\end{align*}
Покажем, что константа $c_{3,0}$ равна нулю. Для этого, предполагая, что $c_{3,0}\neq 0$ введем в рассмотрение несколько операторов из $L'_x$:
\begin{align*}
&R^j=S_1^{j}-\frac{m}{c_5}\left[S_3^{j},S_1^{j}\right]=\sum_kc_{3,0}\frac{\partial}{\partial w^{j}_{k}} \in L'_x, \\
&P^j=m\left[R^{j+1},\left[R^{j+2},S_3^{j}\right]\right]+c_{3,0}\left[R^{j+2},S_3^{j}\right]=c_{3,0}^3\sum_k\tilde{\rho}^{j+2}_k\frac{\partial}{\partial w^{j+1}_{k}},\\
&Q^j=m\left[R^{j-1},\left[R^{j-2},S_3^{j}\right]\right]+c_{3,0}\left[R^{j-2},S_3^{j}\right]=c_{3,0}^3\sum_k\tilde{\rho}^{j-2}_{k+1}\frac{\partial}{\partial w^{j-1}_{k}}.
\end{align*}
Пара операторов $P^{j-1}$, $Q^{j+1}$ в точности совпадает с парой операторов $R^j$, $\bar{R}^j$ (см. \eqref{++}), рассмотренных при доказательстве Леммы~\ref{lem3.1}. Там же было доказано, что последовательность кратных коммутаторов этих операторов линейно независима, если $c_{3,0}\neq0$, что противоречит условию конечномерности алгебры $L_x$. Полученное противоречие убеждает в том, что $c_{3,0}=0$. Лемма доказана.
\end{proof}

Далее мы ищем функции $\tilde A_4(w)$ и $\tilde A_6(w)$, в том случае, когда $\tilde A_3(w)$ имеет вид \eqref{F32}.
Перейдем к новым переменным, полагая 
\begin{align*}
w^{j}_{k}=m\log y^j_k, 
\end{align*}
тогда операторы $S_1^{j}$ и $S_3^{j}$ принимают вид
\begin{align*}
S^j_1=&\sum_k \frac{c_3}{m}\frac{\partial}{\partial y^{j}_{k}},\\
S_3^j=&\sum_k \frac{c_5}{m} y^{j}_{k}\frac{\partial}{\partial y^{j}_{k}}+\bar A_4(y^{j-1}_{k})\frac{\partial}{\partial y^{j-1}_{k}}+\frac{c_3}{m}\bar \rho^{j-2}_{k+1}\frac{\partial}{\partial y^{j-1}_{k}}+\bar A_6(y^{j+1}_{k})\frac{\partial}{\partial y^{j+1}_{k}}+ \frac{c_3}{m}\bar \rho^{j+2}_k \frac{\partial}{\partial y^{j+1}_{k}},
\end{align*}
где $\bar A_i(y)=\frac{y}{m}\tilde{A}_i(m\log y)$, функция $\bar{\rho}^j_i$ получается из \eqref{Trho} заменой $e^{w^{j}_{k}}$ на $(y^j_k)^m$.
Рассматривая две последовательности кратных коммутаторов операторов $S_1^{j}$ и $S_3^{j}$ вида $ad^{s}_{S^j_1}S_3^{j+1}$ и $ad^{s}_{S^j_1}S_3^{j-1}$, где $s\geq1$, можно показать, что пара функций $(\bar A_4(y),\,\bar A_6(y))$ может совпадать лишь с одним из следующих вариантов (1a,\, 1b) либо (2a,\, 2b), где:
\begin{align*}
	&1a) \quad \bar A_4(y)=\bar{c}_{4,0}+\bar{c}_{4,1}y+\bar{c}_{4,2}y^2,\\
	&2a) \quad \bar A_4(y)=\bar{c}_{4,0}+\bar{c}_{4,1}e^{\lambda y}+\bar{c}_{4,2}e^{-\lambda y}, \quad \lambda\neq0\\
	&1b) \quad \bar A_6(y)=\bar{c}_{6,0}+\bar{c}_{6,1}y+\bar{c}_{6,2}y^2,\\
	&2b) \quad \bar A_6(y)=\bar{c}_{6,0}+\bar{c}_{6,1}e^{\lambda y}+\bar{c}_{6,2}e^{-\lambda y}, \quad \lambda\neq0,
\end{align*}
где $\bar{c}_{i,j}$ -- некоторые постоянные.
В случае (2a,\, 2b) параметры $\bar{c}_{4,1}$, $\bar{c}_{4,2}$, $\bar{c}_{6,1}$ и $\bar{c}_{6,2}$ равны нулю. Действительно, воспользовавшись операторами $S_1^{j}$ и $S_3^{j+1}$ можно показать, что операторы вида $Q_\lambda=\sum_k\bar{c}_{4,1}e^{\lambda y^{j}_{k}}\frac{\partial}{\partial y^{j}_{k}}$ и $Q_{-\lambda}=\sum_k\bar{c}_{4,2}e^{-\lambda y^{j}_{k}}\frac{\partial}{\partial y^{j}_{k}}$ лежат в алгебре $L_x$. Исследуем последовательность кратных коммутаторов $S_3^{j}$ и $Q_\lambda$ следующего вида
\begin{align*}
Q_1=\left[S_3^{j},Q_\lambda\right], \quad Q_2=\left[S_3^{j},Q_1\right], \quad Q_3=\left[S_3^{j},Q_2\right], \quad \ldots.
\end{align*}
Легко проверить, что выполняются соотношения:
\begin{align*}
Q_1=\bar{c}_{4,1}\frac{c_5}{m}\lambda\sum_k y^{j}_{k}e^{\lambda y^{j}_{k}}\frac{\partial}{\partial y^{j}_{k}}-\frac{c_5}{m}Q_\lambda, \quad
Q_2=\bar{c}_{4,1}\left(\frac{c_5}{m}\right)^2(\lambda)^2\sum_k (y^{j}_{k})^2e^{\lambda y^{j}_{k}}\frac{\partial}{\partial y^{j}_{k}}-\frac{c_5}{m}Q_1, \quad \ldots.
\end{align*}
По индукции можно показать, что в общем случае справедливо равенство
\begin{align*}
Q_M=\bar{c}_{4,1}\left(\frac{c_5}{m}\right)^M(\lambda)^M\sum_k (y^{j}_{k})^Me^{\lambda y^{j}_{k}}\frac{\partial}{\partial y^{j}_{k}}+\sum_{i=0}^{M-1}r_{M,i}Q_i,
\end{align*}
где $Q_0:=Q_\lambda$, а сомножители $r_{M,i}$~-- некоторые постоянные. Поскольку последовательность операторов
\begin{align*}
\bar{Q}_M=\bar{c}_{4,1}\left(\frac{c_5}{m}\right)^M(\lambda)^M\sum_k (y^{j}_{k})^Me^{\lambda y^{j}_{k}}\frac{\partial}{\partial y^{j}_{k}},
\end{align*}
лежащих в $L_x$, имеет бесконечномерную размерность, то приходим к противоречию, т.к. $\dim L_x<\infty$. Следовательно, наше предположение о том, что $\bar{c}_{4,1}\neq 0$ неверно. 
Поэтому $\bar{c}_{4,1}=0$. 

Аналогично можно показать, что $\bar{c}_{6,1}=0$, $\bar{c}_{4,2}=0$, $\bar{c}_{6,2}=0$. Таким образом, в случае (2a,\, 2b) имеем $\bar A_4(y)=\bar{c}_{4,0}$ и $\bar A_6(y)=\bar{c}_{6,0}$. 

Рассмотрим теперь случай (1a,\, 1b). Вернемся к переменным $w$, полагая $y=e^{\frac{1}{m}w}$. Тогда найденные функции $\bar A_4(y)$, $\bar A_6(y)$ перепишутся в виде
\begin{align*}
	&1a) \quad \tilde{A}_4(w)=c_{4,0}+c_{4,1}e^{-\frac{1}{m}w}+c_{4,2}e^{\frac{1}{m}w},\\
	&1b) \quad \tilde{A}_6(w)=c_{6,0}+c_{6,1}e^{-\frac{1}{m}w}+c_{6,2}e^{\frac{1}{m}w},
\end{align*}
а операторы $S_1^{j}$ и $S_3^{j}$ в терминах $w$ примут вид
\begin{align*}
S^j_1=&\sum_k c_3e^{-\frac{1}{m} w^{j}_{k}}\frac{\partial}{\partial w^{j}_{k}},\\
S_3^j=&\sum_k c_5 \frac{\partial}{\partial w^{j}_{k}}+\tilde A_4(w^{j-1}_{k})\frac{\partial}{\partial w^{j-1}_{k}}+c_3e^{-\frac{1}{m} w^{j-1}_{k}}\tilde \rho^{j-2}_{k+1}\frac{\partial}{\partial w^{j-1}_{k}}+\\
&+\tilde A_6(w^{j+1}_{k})\frac{\partial}{\partial w^{j+1}_{k}}+ c_3e^{-\frac{1}{m} w^{j+1}_{k}}\tilde \rho^{j+2}_k \frac{\partial}{\partial w^{j+1}_{k}}.
\end{align*}
Покажем, что $c_{4,2}=0$. Легко убедиться, что
\begin{align*}
\left[S^{j-1}_1,S^j_3\right]=\frac{c_3}{m}\sum_k\left(2c_{4,2}+c_{4,0}e^{-\frac{1}{m} w^{j-1}_{k}}\right)\frac{\partial}{\partial w^{j-1}_{k}}.
\end{align*}
Отсюда ясно, что при $c_{4,2}\neq 0$ оператор $S^j_0=\sum_k\frac{\partial}{\partial w^{j}_{k}}$ лежит в алгебре $L_x$. Легко можно проверить, что выполняются равенства:
\begin{align*}
&P^j=\left[S^{j-1}_0,S^{j+1}_3\right]=c_3\sum_ke^{-\frac{1}{m} w^{j}_{k}}\tilde \rho^{j-1}_{k+1}\frac{\partial}{\partial w^{j}_{k}},\\
&R^j=\left[S^{j+1}_0,S^{j-1}_3\right]=c_3\sum_ke^{-\frac{1}{m} w^{j}_{k}}\tilde \rho^{j+1}_{k}\frac{\partial}{\partial w^{j}_{k}}.
\end{align*}
Вычитая из $S_3^j$ линейную комбинацию полученных операторов с подходящими значениями $j$ получим, что оператор
\begin{align*}
Q^j=\sum_k c_5 \frac{\partial}{\partial w^{j}_{k}}+\tilde A_4(w^{j-1}_{k})\frac{\partial}{\partial w^{j-1}_{k}}+\tilde A_6(w^{j+1}_{k})\frac{\partial}{\partial w^{j+1}_{k}}
\end{align*}
также лежит в алгебре $L_x$. Построим новый оператор
\begin{align*}
Q_\lambda=\frac{m}{2}\left[S^{j}_0,Q^{j+1}\right]+\frac{m^2}{2}\left[S^{j}_0,\left[S^{j}_0,Q^{j+1}\right]\right]=\sum_kc_{4,2}e^{\frac{1}{m} w^{j}_{k}}\frac{\partial}{\partial w^{j}_{k}}.
\end{align*}
Рассмотрим последовательность кратных коммутаторов вида:
\begin{align*}
Q_\lambda, \quad R^{j-1}, \quad R_1=\left[Q_\lambda,R^{j-1}\right], \quad R_2=\left[Q_\lambda,R_1\right], \quad \ldots.
\end{align*}
Найдем явные представления для этих операторов
\begin{align*}
&R_1=c_3c_{4,2}\sum_ke^{\frac{1}{m} w^{j}_{k}}e^{-\frac{1}{m} w^{j-1}_{k}}\tilde{\rho}^j_k\frac{\partial}{\partial w^{j-1}_{k}}, \\
&R_2=c_3c^2_{4,2}\left(\frac{1}{m}+1\right)\sum_ke^{\frac{2}{m} w^{j}_{k}}e^{-\frac{1}{m} w^{j-1}_{k}}\tilde{\rho}^j_k\frac{\partial}{\partial w^{j-1}_{k}}.
\end{align*}
В общем случае имеем
\begin{align*}
R_M=c_3c^M_{4,2}\left(\frac{1}{m}+1\right)\left(\frac{2}{m}+1\right)\ldots\left(\frac{M-1}{m}+1\right)\sum_ke^{\frac{M}{m} w^{j}_{k}}e^{-\frac{1}{m} w^{j-1}_{k}}\tilde{\rho}^j_k\frac{\partial}{\partial w^{j-1}_{k}}.
\end{align*}
Поскольку $m\geq 1$, то линейная оболочка элементов этой последовательности имеет бесконечную размерность, если $c_{4,2}\neq0$, поэтому $c_{4,2}=0$. Аналогично можно показать, что $c_{6,2}=0$. Следовательно, в случае (1a,\, 1b) функции $\tilde{A}_4(w)$ и $\tilde{A}_6(w)$ имеют вид:
\begin{align*}
	&\tilde{A}_4(w)=c_{4,0}+c_{4,1}e^{-\frac{1}{m}w},\\
	&\tilde{A}_6(w)=c_{6,0}+c_{6,1}e^{-\frac{1}{m}w}.
\end{align*}

Подведем итог рассуждениям выше, сформулируем его в виде теоремы. 

\begin{theorem}\label{final}
Если цепочка \eqref{jnx-part}, \eqref{form-f} интегрируема в смысле определения~\ref{def}, то функция \eqref{form-f}
\begin{align}\label{2*}
f=&A_3(\tau)u^{j+1}_nu^{j-1}_{n+1}+A_4(\tau)u^{j+1}_n+A_5(\tau)u^j_n+A_6(\tau)u^{j-1}_{n+1}+A_7(\tau)
\end{align}
принадлежит одному из следующих трех классов
\begin{align*}
1)\, &A_3=0, \quad A_4=c_{4,0}\tau+c_{4,1}\tau \log \tau+c_{4,2}\tau (\log \tau)^2, \quad A_6=c_{6,0}\tau+c_{6,1}\tau \log \tau+c_{6,2}\tau (\log \tau)^2,\\
2)\, &A_3=0, \quad A_4=c_{4,0}\tau+c_{4,1}\tau^{1+\lambda}+c_{4,2}\tau^{1-\lambda}, \quad A_6=c_{6,0}\tau+c_{6,1}\tau^{1+\lambda}+c_{6,2}\tau^{1-\lambda},\\
3)\, &A_3=c_3\tau^{1-\frac{1}{m}}, \quad A_4=c_{4,0}\tau+c_{4,1}\tau^{1-\frac{1}{m}}, \quad A_6=c_{6,0}\tau+c_{6,1}\tau^{1-\frac{1}{m}}.
\end{align*}
где $\tau=u^0_n-u^{0}_{n+1}$, коэффициенты $A_5$, $A_7$ имеют вид $A_5(\tau)=c_5\tau$, $A_7(\tau)=-\frac{c_5}{2}\tau^2+c_7\tau$. 
\end{theorem}
В результате рассматриваемая классификационная задача свелась к уточнению значений постоянных параметров в формулах 1) -- 3) Теоремы~\ref{final}. 

\section{Дальнейший анализ цепочек класса 1)}

Цель настоящего раздела состоит в дальнейшем уточнении вида искомой функции \eqref{2*}, принадлежащей классу 1) Теоремы \ref{final}. Для этого мы воспользуемся наиболее простой нетривиальной редукцией цепочки \eqref{jnx-part}, \eqref{form-f} размерности два
\begin{align*}
&u^0_{n+1,x}=u^0_{n,x}+f(u^{1}_{n},u^{0}_n,u^0_{n+1}),\\
&u^1_{n+1,x}=u^1_{n,x}+f(u^{1}_n,u^1_{n+1},u^{0}_{n+1}),
\end{align*}
полученной в результате наложения на цепочку вырожденных условий обрыва $u^{-1}_n\equiv 0$, $u^{2}_n\equiv 0$. 
Эта система в силу \eqref{2*} очевидно принимает вид
\begin{align*}
&-\tau^0_{n,x}=A_4(\tau^0_n)\left(u^{1}_0+\rho^1_n\right)+A_5(\tau^0_n)\left(u^{0}_0+\rho^0_n\right)+A_7(\tau^0_n),\\
&-\tau^1_{n,x}=A_5(\tau^1_n)\left(u^{1}_0+\rho^1_n\right)+A_6(\tau^1_n)\left(u^{0}_0+\rho^0_{n+1}\right)+A_7(\tau^1_n),
\end{align*}
где $\tau^0_n=u^{0}_n-u^0_{n+1}$, $\tau^1_n=u^{1}_n-u^1_{n+1}$, $\rho^i_k$ задана формулой \eqref{Trho}. 

В качестве классификационного критерия, как и выше, берется условие конечномерности характеристической алгебры Ли-Райнхарта $L_x$, которая в этом случае порождается следующими тремя операторами:   
 $\bar{X}_0=\frac{\partial}{\partial u^0_0}$, $\bar{X}_1=\frac{\partial}{\partial u^1_0}$ и
\begin{align*}
Y=&\sum_{n}\left(A_4(\tau^0_n)\left(u^{1}_0+\rho^1_n\right)+A_5(\tau^0_n)\left(u^{0}_0+\rho^0_n\right)+A_7(\tau^0_n)\right)\frac{\partial}{\partial \tau^0_n}+\\
&+\sum_{n}\left(A_5(\tau^1_n)\left(u^{1}_0+\rho^1_n\right)+A_6(\tau^1_n)\left(u^{0}_0+\rho^0_{n+1}\right)+A_7(\tau^1_n)\right)\frac{\partial}{\partial \tau^1_n}.
\end{align*}
Легко заметить, что из конечномерности алгебры $L_x$ следует конечномерность алгебры $L'_x$, порожденной операторами
\begin{align*}
S_0=&\left[\bar{X}_0,Y\right]=\sum_nA_5(\tau^0_n)\frac{\partial}{\partial \tau^0_n}+\sum_nA_6(\tau^1_n)\frac{\partial}{\partial \tau^1_n},\\
S_1=&\left[\bar{X}_1,Y\right]=\sum_nA_4(\tau^0_n)\frac{\partial}{\partial \tau^0_n}+\sum_nA_5(\tau^1_n)\frac{\partial}{\partial \tau^1_n},\\
S=&Y-u^0_0S_0-u^1_0S_1=\sum^{+\infty}_{n=-\infty}\left(A_4(\tau^0_n)\rho^1_n+A_5(\tau^0_n)\rho^0_n+A_7(\tau^0_n)\right)\frac{\partial}{\partial \tau^0_n}+\\
&+\sum_{n}\left(A_5(\tau^1_n)\rho^1_n+A_6(\tau^1_n)\rho^0_{n+1}+A_7(\tau^1_n)\right)\frac{\partial}{\partial \tau^1_n}.
\end{align*}
Для удобства перейдем в этих операторах к переменным $w=\ln \tau$:
\begin{align*}
&S_0=\sum_n\tilde{A}_5(w^0_n)\frac{\partial}{\partial w^0_n}+\sum_n\tilde{A}_6(w^1_n)\frac{\partial}{\partial w^1_n},\qquad
S_1=\sum_n\tilde{A}_4(w^0_n)\frac{\partial}{\partial w^0_n}+\sum_n\tilde{A}_5(w^1_n)\frac{\partial}{\partial w^1_n},\\
&S=\sum_{n}\left(\tilde{A}_4(w^0_n)\tilde{\rho}^1_n+\tilde{A}_5(w^0_n)\tilde{\rho}^0_n+A_7(w^0_n)\right)\frac{\partial}{\partial w^0_n}+\\
&\qquad+\sum_{n}\left(\tilde{A}_5(w^1_n)\tilde{\rho}^1_n+\tilde{A}_6(w^1_n)\tilde{\rho}^0_{n+1}+A_7(w^1_n)\right)\frac{\partial}{\partial w^1_n}.
\end{align*}
\begin{lemma}\label{important}
Если нарушается хотя бы одно из следующих условий
\begin{align}\label{cond-ci}
c_{4,2}\left(4c_{4,0}c_{4,2}-c^2_{4,1}\right)=0, \qquad c_{6,2}\left(4c_{6,0}c_{6,2}-c^2_{6,1}\right)=0,
\end{align}
то $\dim L_x=\infty$.
\end{lemma}
\begin{proof} Учитывая, что функция $f$ принадлежит классу 1) теоремы~\ref{final}, найдем кратные коммутаторы
\begin{align*}
&ad^3_{S_1}S_0=c_5\sum_n\left(4c_{4,0}c_{4,2}-c^2_{4,1}\right)\left(2c_{4,2}w^0_n+c_{4,1}\right)\frac{\partial}{\partial w^0_n},\\
&ad_{S_0}ad^3_{S_1}S_0=c^2_5\sum_n2c_{4,2}\left(4c_{4,0}c_{4,2}-c^2_{4,1}\right)\frac{\partial}{\partial w^0_n}.
\end{align*}
Если нарушается первое условие леммы, т.е. $c_{4,2}\left(4c_{4,0}c_{4,2}-c^2_{4,1}\right)\neq 0$, то операторы $Q_1:=\sum_n\frac{\partial}{\partial w^0_n}$ и $Q_2:=\sum_nw^0_n\frac{\partial}{\partial w^0_n}$ лежат в $L_x$. Тогда, как легко заметить, операторы $ad^m_{Q_2}S$ допускают простое явное представление. Для $m=1,2$ имеем
\begin{align*}
&ad_{Q_2}S=\sum_nF^1_n\frac{\partial}{\partial w^0_n}+\sum_n\tilde{A}_6(w^1_n)\tilde{\rho}^0_{n+1}w^0_n\frac{\partial}{\partial w^1_n},\\
&ad^2_{Q_2}S=\sum_nF^2_n\frac{\partial}{\partial w^0_n}+\sum_n\tilde{A}_6(w^1_n)\tilde{\rho}^0_{n+1}\left((w^0_n)^2+w^0_n\right)\frac{\partial}{\partial w^1_n},
\end{align*}
а при $m\geq 3$:
\begin{align*}
ad^2_{Q_2}S=\sum_nF^m_n\frac{\partial}{\partial w^0_n}+\sum_n\tilde{A}_6(w^1_n)\tilde{\rho}^0_{n+1}\left((w^0_n)^m+c_{m-1}(w^0_n)^{m-1}+\ldots+c_1w^0_n\right)\frac{\partial}{\partial w^1_n},
\end{align*}
где $F_n^m$~-- некоторая функция, явный вид которой мы не уточняем. Ясно, что линейная оболочка семейства операторов $ad^m_{Q_2}S$ имеет бесконечную размерность, поэтому $\dim L_x=\infty$. Вторая часть леммы доказывается аналогично.
\end{proof}

\begin{lemma}\label{lem-ci}
Пусть выполнены условия \eqref{cond-ci}. Тогда для конечномерности алгебры $L_x$, порожденной операторами $S_0$, $S_1$, $S$ необходимо, чтобы параметры $c_{4,1}$, $c_{4,2}$, $c_{6,1}$ и $c_{6,2}$ удовлетворяли равенствам:
\begin{align*}
c_{4,1}=c_{4,2}=c_{6,1}=c_{6,2}=0.
\end{align*}
\end{lemma}

{\it Схема доказательства.} При выполнении \eqref{cond-ci} возникают семь непересекающихся наборов условий:
\begin{enumerate}
  \item[1i)] $c_{4,2}=0, \quad c_{6,2}=0, \quad 4c_{4,0}c_{4,2}= c^2_{4,1}, \quad 4c_{6,0}c_{6,2}= c^2_{6,1}$;
	\item[2i)] $c_{4,2}=0, \quad c_{6,2}=0, \quad 4c_{4,0}c_{4,2}= c^2_{4,1}, \quad 4c_{6,0}c_{6,2}\neq c^2_{6,1}$;
	\item[3i)] $c_{4,2}=0, \quad c_{6,2}=0, \quad 4c_{4,0}c_{4,2}\neq c^2_{4,1}, \quad 4c_{6,0}c_{6,2}= c^2_{6,1}$;
	\item[4i)] $c_{4,2}=0, \quad c_{6,2}=0, \quad 4c_{4,0}c_{4,2}\neq c^2_{4,1}, \quad 4c_{6,0}c_{6,2}\neq c^2_{6,1}$;
	\item[5i)] $c_{4,2}=0, \quad c_{6,2}\neq0, \quad 4c_{4,0}c_{4,2}\neq c^2_{4,1}, \quad 4c_{6,0}c_{6,2}= c^2_{6,1}$;
	\item[6i)] $c_{4,2}\neq0, \quad c_{6,2}=0, \quad 4c_{4,0}c_{4,2}=c^2_{4,1}, \quad 4c_{6,0}c_{6,2}\neq c^2_{6,1}$;
	\item[7i)] $c_{4,2}\neq0, \quad c_{6,2}\neq0, \quad 4c_{4,0}c_{4,2}= c^2_{4,1}, \quad 4c_{6,0}c_{6,2}= c^2_{6,1}$.
\end{enumerate}
Из этих вариантов в действительности реализуется только вариант $1i)$. Во всех остальных случаях характеристическая алгебра бесконечномерна. Рассмотрение случаев $2i)-6i)$ вполне аналогично доказательству Леммы~\ref{important}. Рассмотрение случая $7i)$ несколько отличается. Здесь мы пользуемся последовательностью кратных коммутаторов вида
\begin{align*}
S_0, \quad S, \quad [S_0,S], \quad [S_0,[S_0,S]], \quad [S_0,[S_0,[S_0,S]]], \quad \ldots,
\end{align*}
которая содержит бесконечное множество линейно независимых элементов.

Ниже нам понадобится следующий дискретный вариант леммы Шабата (см. книгу~\cite{ZMHS})
\begin{lemma}\label{DS}
Предположим, что векторное поле вида
\begin{align*}
K=\sum_{j=1}^N\sum_{n} a^j_n \frac{\partial}{\partial w^j_n},
\end{align*}
где для всех $j=1,2,\dots N$ выполнено условие $ a^j_0=0$, удовлетворяет соотношению
\begin{align*}
D_n K D_n^{-1}=hK
\end{align*}
с некоторым множителем $h$, тогда $K=0$.
\end{lemma}
Доказательство Леммы \ref{DS} можно найти в работе \cite{HabKh21}.
Имеет место следующая теорема.
\begin{theorem}\label{3.2}
Предположим, что цепочка \eqref{jnx-part}, \eqref{form-f} интегрируема в смысле определения \ref{def} и принадлежит классу 1) Теоремы \ref{final}. Тогда если выполнено условие 
\begin{align}\label{Yc}
(c_5-c_{4,0})(c_5-c_{6,0})\neq 0,
\end{align}
то цепочка точечными заменами приводится к виду
\begin{align}\label{F01}
u^j_{n+1,x}= u^j_{n,x}+(u^j_n-u^j_{n+1})(u^j_n+u^j_{n+1} -Mu^{j+1}_n-\bar M u^{j-1}_{n+1}),
\end{align}
где $M\geq1$ и $\bar M\geq1$ некоторые целые числа. А если это условие нарушается, то цепочка заменой приводится к виду
\begin{align}\label{F02}
u^j_{n+1,x}= u^j_{n,x}+(u^j_n-u^j_{n+1})\left(u^j_n+u^j_{n+1}+C_{4}u^{j+1}_n+C_{6} u^{j-1}_{n+1}\right),
\end{align}
где $C_4=\frac{2c_{4,0}}{c_5}$, $C_6=\frac{2c_{6,0}}{c_5}$ произвольные постоянные, при этом выполняется хотя бы одно из следующих двух условий: 1) $C_4=2$, 2) $C_6=2$.
\end{theorem}

\begin{proof}
Пусть сначала выполняется условие \eqref{Yc}. Уточним постоянные $c_{4,0}$ и $c_{6,0}$. В силу Леммы~\ref{lem-ci} операторы $S_0$, $S_1$, $S$ принимают вид:
\begin{align*}
&S_0=\sum_nc_5\frac{\partial}{\partial w^0_n}+\sum_nc_{6,0}\frac{\partial}{\partial w^1_n},\qquad
S_1=\sum_nc_{4,0}\frac{\partial}{\partial w^0_n}+\sum_nc_5\frac{\partial}{\partial w^1_n},\\
&S=\sum^{+\infty}_{n=-\infty}\left(c_{4,0}\tilde{\rho}^1_n+c_5\tilde{\rho}^0_n-\frac{c_5}{2}e^{w^0_n}+c_7\right)\frac{\partial}{\partial w^0_n}+\sum_{n}\left(c_5\tilde{\rho}^1_n+c_{6,0}\tilde{\rho}^0_{n+1}-\frac{c_5}{2}e^{w^1_n}+c_7\right)\frac{\partial}{\partial w^1_n}.
\end{align*}
Заметим, что при выполнении условия \eqref{Yc} можно рассматривать операторы вида 
\begin{align*}
&P=\frac{1}{c_{4,0}(c_5-c_{4,0})}\left(c_5[S_1,S]-[S_1,[S_1,S]]\right), \\
&Q=\frac{1}{c_{6,0}(c_5-c_{6,0})}\left(c_5[S_0,S]-[S_0,[S_0,S]]\right),
\end{align*}
поскольку постоянные  $c_{4,0}$ и $c_{6,0}$ отличны от нуля по условию задачи (см. \eqref{jnx-part}). Операторы $P$ и $Q$ удобны тем, что они допускают сравнительно простое 
 координатное представление
\begin{align*}
&P=c_5\sum_n\left(\tilde{\rho}^0_n-\frac{1}{2}e^{w^0_n}\right)\frac{\partial}{\partial w^0_n}+c_{6,0}\sum_n\tilde{\rho}^0_{n+1}\frac{\partial}{\partial w^1_n}, \\
&Q=c_{4,0}\sum_n\tilde{\rho}^1_{n}\frac{\partial}{\partial w^0_n}+c_5\sum_n\left(\tilde{\rho}^1_n-\frac{1}{2}e^{w^1_n}\right)\frac{\partial}{\partial w^1_n}.
\end{align*}

Рассмотрим следующую бесконечную последовательность  операторов
\begin{align*}
Q, \quad P, \quad Q_1=[Q,P], \quad Q_2=[Q,Q_1], \quad Q_3=[Q,Q_2], \quad \ldots.
\end{align*}
Выясним действие автоморфизма \eqref{auto-main} на элементы этой последовательности:
\begin{align*}
D_nQD^{-1}_n=&Q+e^{w^1_0}S_1, \quad D_nPD^{-1}_n=P+e^{w^0_0}S_0, \quad D_nS_0D^{-1}_n=S_0, \quad D_nS_1D^{-1}_n=S_1,\\
D_nQ_1D^{-1}_n=&Q_1+c_{4,0}e^{w^1_0}P-c_{6,0}e^{w^0_0}Q+c_{4,0}e^{w^1_0+w^0_0}S_0,\\
D_nQ_2D^{-1}_n=&Q_2+(c_5+2c_{4,0})e^{w^1_0}Q_1+c_{4,0}\left(\frac{c_5}{2}+c_{4,0}\right)e^{2w^1_0}P-\\
&-c_{6,0}(c_5+2c_{4,0})e^{w^1_0+w^0_0}Q+c_{4,0}\left(\frac{c_5}{2}+c_{4,0}\right)e^{2w^1_0+w^0_0}S_0,\\
D_nQ_3D^{-1}_n=&Q_3+3(c_5+c_{4,0})e^{w^1_0}Q_2+3\left(\frac{c_5}{2}+c_{4,0}\right)(c_5+c_{4,0})e^{2w^1_0}Q_1+\\
&+c_{4,0}\left(\frac{c_5}{2}+c_{4,0}\right)(c_5+c_{4,0})e^{3w^1_0}P-3c_{6,0}\left(\frac{c_5}{2}+c_{4,0}\right)(c_5+c_{4,0})e^{2w^1_0+w^0_0}Q+\\
&+c_{4,0}\left(\frac{c_5}{2}+c_{4,0}\right)(c_5+c_{4,0})e^{3w^1_0+w^0_0}S_0.
\end{align*}
Можно доказать по индукции, что в общем случае $m\geq2$ автоморфизм алгебры действует так
\begin{align*}
D_nQ_mD^{-1}_n=&Q_m+\alpha_{m,1}e^{w^1_0}Q_{m-1}+\alpha_{m,2}e^{2w^1_0}Q_{m-2}+\ldots+\alpha_{m,m-1}e^{(m-1)w^1_0}Q_{1}+\\
&+\alpha_{m,m}e^{mw^1_0}P+\alpha_{m,m+1}e^{(m-1)w^1_0+w^0_0}Q+\alpha_{m,m+2}e^{mw^1_0+w^0_0}S_0,
\end{align*}
где все коэффициенты $\alpha_{m,j}$ постоянны. При этом первый из них имеет вид
\begin{align}\label{alpham1}
\alpha_{m,1}=m\left((m-1)\frac{c_5}{2}+c_{4,0}\right).
\end{align}
В силу конечномерности алгебры $L_x$ существует такой номер $M$, что операторы $Q_M$, $Q_{M-1}$, \ldots, $Q_1$, $Q$, $P$, $S_0$, $S_1$~-- линейно независимы, а оператор $Q_{M+1}$ выражается в виде
\begin{align}\label{PQm}
Q_{M+1}=\lambda_1Q_{M}+\lambda_2Q_{M-1}+\ldots+\lambda_MQ_{1}+\lambda_{M+1}P+\lambda_{M+2}Q+\lambda_{M+3}S_0+\lambda_{M+4}S_1.
\end{align}
Применив к равенству \eqref{PQm} автоморфизм получим
\begin{align}\label{eqQm}
\begin{aligned}
\lambda_1Q_{M}&+\lambda_2Q_{M-1}+\lambda_3Q_{M-2}+\ldots+\lambda_{M+4}S_1+\alpha_{M+1,1}e^{w^1_0}Q_{M}+\\
&+\alpha_{M+1,2}e^{2w^1_0}Q_{M-1}+\ldots+\alpha_{M+1,M+3}e^{(M+1)w^1_0+w^0_0}S_0=\\
=&D_n(\lambda_1)\left(Q_M+\alpha_{M,1}e^{w^1_0}Q_{M-1}+\alpha_{M,2}e^{2w^1_0}Q_{M-2}+\ldots\right)+\\
&+D_n(\lambda_2)\left(Q_{M-1}+\alpha_{M-1,1}e^{w^1_0}Q_{M-2}+\ldots\right)+\ldots+D_n(\lambda_{M+4})S_1.
\end{aligned}
\end{align}
Сравнивая коэффициенты при независимых операторах получим систему уравнений для определения коэффициентов $\lambda_i$. Имеем:
\begin{align*}
Q_{M}: \quad D_n(\lambda_1)-\lambda_1=\alpha_{M+1,1}e^{w^1_0},
\end{align*}
откуда сразу вытекает, что $\lambda_1=\const$, $\alpha_{M+1,1}=0$, или в силу \eqref{alpham1}
\begin{align}\label{c40}
c_{4,0}=-\frac{1}{2}c_5M, \quad M\geq 1.
\end{align}
Далее находим
\begin{align*}
Q_{M-1}: \quad D_n(\lambda_2)-\lambda_2=\alpha_{M+1,2}e^{2w^1_0}-\lambda_1\alpha_{M,1}e^{w^1_0}.
\end{align*}
Отсюда имеем $\lambda_2=\const$, $\alpha_{M+1,2}=0$, $\lambda_1\alpha_{M,1}=0$.
\begin{align*}
Q_{M-2}: \quad D_n(\lambda_3)-\lambda_3=\alpha_{M+1,3}e^{3w^1_0}-\lambda_1\alpha_{M,2}e^{2w^1_0}-\lambda_2\alpha_{M-1,1}e^{w^1_0}.
\end{align*}
В итоге находим $\lambda_3=\const$, $\alpha_{M+1,3}=0$, $\lambda_1\alpha_{M,2}=0$, $\lambda_2\alpha_{M-1,1}=0$.

Далее сравнивая коэффициенты при всех остальных независимых операторах получаем, что для $\forall i$ $\alpha_{M+1,i}=0$, $\lambda_1\alpha_{M,i}=0$, $\lambda_k\alpha_{M-k+1,i}=0$. Отсюда следует, что 
\begin{align}\label{autoQM1}
D_nQ_{M+1}D_n^{-1}=Q_{M+1}.
\end{align}
Покажем, что из равенства \eqref{autoQM1} следует, что $Q_{M+1}=0$. С этой целью более детально исследуем явное координатное представление операторов последовательности. Для нескольких первых элементов  имеем
\begin{align*}
Q_1=&c_{4,0}\sum_n\tilde{\rho}^1_n\left(c_5\tilde{\rho}^0_n-\frac{c_5}{2}e^{w^0_n}-c_{6,0}\tilde{\rho}^0_{n+1}\right)\frac{\partial}{\partial w^0_n}+c_{6,0}\sum_n\tilde{\rho}^0_{n+1}\left(c_{4,0}\tilde{\rho}^1_n-c_5\tilde{\rho}^1_n+\frac{c_5}{2}e^{w^1_n}\right)\frac{\partial}{\partial w^1_n},\\
Q_2=&c_{4,0}\sum_n\tilde{\rho}^1_n\left(\left(c_5\tilde{\rho}^0_n-\frac{c_5}{2}e^{w^0_n}-c_{6,0}\tilde{\rho}^0_{n+1}\right)\left(c_{4,0}\tilde{\rho}^1_n+c_5\tilde{\rho}^1_n-\frac{c_5}{2}e^{w^1_n}\right)\right.-\\&\left.-c_{6,0}\tilde{\rho}^0_{n+1}\left(c_{4,0}\tilde{\rho}^1_n-c_5\tilde{\rho}^1_n+\frac{c_5}{2}e^{w^1_n}\right)\right)\frac{\partial}{\partial w^0_n}+\\
&+c_{4,0}c_{6,0}\sum_n\tilde{\rho}^1_n\tilde{\rho}^0_{n+1}\left(c_{4,0}\tilde{\rho}^1_n-c_5\tilde{\rho}^1_n+\frac{c_5}{2}e^{w^1_n}\right)\frac{\partial}{\partial w^1_n},
\end{align*}
Для общего случая $m\geq2$ справедливо представление
\begin{align*}
Q_m=&\sum_n\tilde{\rho}^1_n\left(F_n^m\frac{\partial}{\partial w^0_n} + G_n^m\frac{\partial}{\partial w^1_n}\right)
\end{align*}
где $F^m_n$, $G^m_n$~-- некоторая гладкие функции. При $m=2$ это легко следует из явной формулы. При $m>2$ можно доказать по индукции пользуясь тем, что оператор $\frac{\partial}{\partial w^j_n}$ переводит $\rho^1_n$ либо в ноль либо в $\rho^1_n$.

Характерная деталь всех этих формул, начиная с $Q_2$ состоит в том, что коэффициенты при операторах дифференцирования вида $\frac{\partial}{\partial w^0_0}$ и  $\frac{\partial}{\partial w^1_0}$ равны нулю, поскольку $\rho^1_0=0$. Тогда в силу Леммы \ref{DS} из равенства \eqref{autoQM1} немедленно вытекает, что $Q_{M+1}=0$, при $M\geq 2$. Следовательно, линейное пространство над кольцом функций от динамических переменных, натянутое на элементы рассматриваемой последовательности имеет конечную размерность тогда и только тогда, когда выполняется равенство \eqref{c40}, выражающее связь между параметрами $c_{4,0}$ и $c_5$. Применяя аналогичные рассуждения к последовательности операторов
\begin{align*}
Q, \quad P, \quad P_1=[P,Q], \quad P_2=[P,P_1], \quad P_3=[P,P_2], \quad \ldots.
\end{align*}
можно показать, что из конечномерности алгебры $L_x$ вытекает равенство
\begin{align}\label{c60}
c_{6,0}=-\frac{1}{2}c_5\bar M, \quad \bar M\geq 1.
\end{align}
Формулы \eqref{c40} и \eqref{c60} устанавливают связь между константами цепочки \eqref{jnx-part}, \eqref{2*}. В результате этого цепочка существенно упрощается и при помощи замен переменных легко приводится к виду \eqref{F01}.

Пусть теперь условие \eqref{Yc} нарушается. Тогда цепочка имеет вид:
\begin{align*}
u^j_{n+1,x}= u^j_{n,x}+(u^j_n-u^j_{n+1})\left(\frac{1}{2}c_5u^j_n+\frac{1}{2}c_5u^j_{n+1}+c_{4,0}u^{j+1}_n+c_{6,0} u^{j-1}_{n+1}\right)+c_7(u^j_n-u^j_{n+1}),
\end{align*}
где предполагается, что выполняется хотя бы одно из условий $c_{4,0}=c_5$, $c_{6,0}=c_5$.
Ясно, что точечными заменами эту цепочку можно привести к виду \eqref{F02}.
\end{proof}

\section{Обсуждение конкретных примеров}

Найденные в Теореме \ref{final} классы 2) и 3), а также цепочки \eqref{F01} и \eqref{F02}, полученные в результате дальнейшей работы с классом 1) нуждаются в дополнительном исследовании с привлечением характеристической алгебры по дискретному направлению. Однако уже на этом этапе нам удалось обнаружить довольно интересные примеры.

{\bf Пример 1.} Полагая $C_{4}=C_{6}=2$ в \eqref{F02} приходим к цепочке
\begin{align}\label{F21}
u^j_{n+1,x}= u^j_{n,x}+(u^j_n-u^j_{n+1})(u^j_n+u^j_{n+1}+2u^{j+1}_n+2u^{j-1}_{n+1}).
\end{align}
А выбирая $M=\bar M=2$ в \eqref{F01} получаем очень близкую к ней цепочку
\begin{align}\label{F12}
u^j_{n+1,x}= u^j_{n,x}+(u^j_n-u^j_{n+1})(u^j_n+u^j_{n+1}-2u^{j+1}_n-2u^{j-1}_{n+1}).
\end{align}
Для обеих цепочек редукции имеют конечномерные характеристические алгебры по направлению $x$. Для одномерной редукции $u^j_{n+1,x}= u^j_{n,x}+(u^j_n)^2-(u^j_{n+1})^2$ это следует из результатов работы \cite{Adler}. Для редукций в виде систем двух и трех уравнений конечномерность алгебры $L_x$ была проверена М.Н.~Кузнецовой\footnote{Частное сообщение}. Вопрос о свойствах характеристической алгебры по направлению $n$ для этих редукций остается открытым.

{\bf Пример 2.} Наиболее интересной является цепочка соответствующая выбору $M=\bar M=1$ в представлении \eqref{F01} 
\begin{align}\label{F31}
u^j_{n+1,x}= u^j_{n,x}+(u^j_n-u^j_{n+1})(u^j_n+u^j_{n+1}-u^{j+1}_n-u^{j-1}_{n+1}).
\end{align}
Эта цепочка полностью проходит наш тест как по направлению $x$, так и по направлению $n$, т.е. она является интегрируемой в смысле определения \eqref{def}. Цепочка \eqref{F31} является интегрируемой и в общепринятом смысле, что подтверждается наличием пары Лакса (см. \cite{HKhS23}).

\section*{Благодарности} Исследование выполнено за счет гранта Российского научного фонда №21-11-00006, https://rscf.ru/project/21-11-00006/.


Исмагил Талгатович Хабибуллин (ответственный за переписку)\\
Институт математики с ВЦ УФИЦ РАН,\\ ул. Чернышевского, 112,\\ 450008, г.Уфа, Россия \\ 
электронная почта {habibullinismagil@gmail.com}\\

Айгуль Ринатовна Хакимова\\
Институт математики с ВЦ УФИЦ РАН,\\ ул. Чернышевского, 112,\\ 450008, г.Уфа, Россия \\ 
электронная почта {aigul.khakimova@mail.ru}\\

\end{document}